\documentclass[12pt,american]{article}
\usepackage[T1]{fontenc}
\usepackage{geometry}
\usepackage{color}
\usepackage{babel}
\usepackage{array}
\usepackage{booktabs}
\usepackage{mathrsfs}
\usepackage{enumitem}
\usepackage{bm}
\usepackage{multirow}
\usepackage{amsmath}
\usepackage{amsthm}
\usepackage{amssymb}
\usepackage{graphicx}
\usepackage{setspace}
\usepackage[authoryear]{natbib}
\usepackage[unicode=true,pdfusetitle,
 bookmarks=true,bookmarksnumbered=false,bookmarksopen=false,
 breaklinks=false,pdfborder={0 0 1},backref=false,colorlinks=false]
 {hyperref}
\hypersetup{
 colorlinks,linkcolor={blue!50!black},citecolor={blue!50!black},urlcolor={blue}}

\makeatletter

\providecommand{\tabularnewline}{\\}

\theoremstyle{plain}
\newtheorem{assumption}{\protect\assumptionname}
\theoremstyle{plain}
\newtheorem{prop}{\protect\propositionname}
\theoremstyle{plain}
\newtheorem{thm}{\protect\theoremname}
\theoremstyle{definition}
\newtheorem{defn}{\protect\definitionname}
\theoremstyle{definition}
 \newtheorem{example}{\protect\examplename}
\theoremstyle{plain}
\newtheorem{lem}{\protect\lemmaname}

\usepackage{babel}
\usepackage{lmodern}
\usepackage{mathrsfs}
\usepackage{breakurl}
\usepackage{bbm}
\usepackage{tikz}
\usepackage{tikz-3dplot}
\usetikzlibrary{calc,fadings,decorations.pathreplacing}

\allowdisplaybreaks
\tikzset{%
  >=latex, 
  inner sep=0pt,%
  outer sep=2pt,%
  mark coordinate/.style={inner sep=0pt,outer sep=0pt,minimum size=3pt,
    fill=black,circle}%
}

\makeatletter
\newcommand{\customlabel}[2]{%
   \protected@write \@auxout {}{\string \newlabel {#1}{{#2}{\thepage}{#2}{#1}{}} }%
   \hypertarget{#1}{}
}
\makeatother

\usepackage[nolists]{endfloat}

\makeatother

\providecommand{\assumptionname}{Assumption}
\providecommand{\definitionname}{Definition}
\providecommand{\examplename}{Example}
\providecommand{\lemmaname}{Lemma}
\providecommand{\propositionname}{Proposition}
\providecommand{\theoremname}{Theorem}

\begin{document}
\global\long\def\a{\alpha}%
 
\global\long\def\b{\beta}%
 
\global\long\def\g{\gamma}%
 
\global\long\def\d{\delta}%
 
\global\long\def\e{\epsilon}%
 
\global\long\def\l{\lambda}%
 
\global\long\def\t{\theta}%
 
\global\long\def\o{\omega}%
 
\global\long\def\s{\sigma}%

\global\long\def\G{\Gamma}%
 
\global\long\def\D{\Delta}%
 
\global\long\def\L{\Lambda}%
 
\global\long\def\T{\Theta}%
 
\global\long\def\O{\Omega}%
 
\global\long\def\R{\mathbb{R}}%
 
\global\long\def\N{\mathbb{N}}%
 
\global\long\def\Q{\mathbb{Q}}%
 
\global\long\def\I{\mathbb{I}}%
 
\global\long\def\P{\mathbb{P}}%
 
\global\long\def\E{\mathbb{E}}%
\global\long\def\B{\mathbb{\mathbb{B}}}%
\global\long\def\S{\mathbb{\mathbb{S}}}%
\global\long\def\V{\mathbb{\mathbb{V}}\text{ar}}%

\global\long\def\X{{\bf X}}%
\global\long\def\cX{\mathscr{X}}%
 
\global\long\def\cY{\mathscr{Y}}%
 
\global\long\def\cA{\mathscr{A}}%
 
\global\long\def\cB{\mathscr{B}}%
 
\global\long\def\cM{\mathscr{M}}%
\global\long\def\cN{\mathcal{N}}%
\global\long\def\cG{\mathcal{G}}%
\global\long\def\cC{\mathcal{C}}%
\global\long\def\sp{\,}%

\global\long\def\es{\emptyset}%
 
\global\long\def\mc#1{\mathscr{#1}}%
 
\global\long\def\ind{\mathbf{\mathbbm1}}%
\global\long\def\indep{\perp}%

\global\long\def\any{\forall}%
 
\global\long\def\ex{\exists}%
 
\global\long\def\p{\partial}%
 
\global\long\def\cd{\cdot}%
 
\global\long\def\Dif{\nabla}%
 
\global\long\def\imp{\Rightarrow}%
 
\global\long\def\iff{\Leftrightarrow}%

\global\long\def\up{\uparrow}%
 
\global\long\def\down{\downarrow}%
 
\global\long\def\arrow{\rightarrow}%
 
\global\long\def\rlarrow{\leftrightarrow}%
 
\global\long\def\lrarrow{\leftrightarrow}%

\global\long\def\abs#1{\left|#1\right|}%
 
\global\long\def\norm#1{\left\Vert #1\right\Vert }%
 
\global\long\def\rest#1{\left.#1\right|}%

\global\long\def\bracket#1#2{\left\langle #1\middle\vert#2\right\rangle }%
 
\global\long\def\sandvich#1#2#3{\left\langle #1\middle\vert#2\middle\vert#3\right\rangle }%
 
\global\long\def\turd#1{\frac{#1}{3}}%
 
\global\long\def\ellipsis{\textellipsis}%
 
\global\long\def\sand#1{\left\lceil #1\right\vert }%
 
\global\long\def\wich#1{\left\vert #1\right\rfloor }%
 
\global\long\def\sandwich#1#2#3{\left\lceil #1\middle\vert#2\middle\vert#3\right\rfloor }%

\global\long\def\abs#1{\left|#1\right|}%
 
\global\long\def\norm#1{\left\Vert #1\right\Vert }%
 
\global\long\def\rest#1{\left.#1\right|}%
 
\global\long\def\inprod#1{\left\langle #1\right\rangle }%
 
\global\long\def\ol#1{\overline{#1}}%
 
\global\long\def\ul#1{\underline{#1}}%
 
\global\long\def\td#1{\tilde{#1}}%
\global\long\def\bs#1{\boldsymbol{#1}}%

\global\long\def\upto{\nearrow}%
 
\global\long\def\downto{\searrow}%
 
\global\long\def\pto{\overset{p}{\longrightarrow}}%
 
\global\long\def\dto{\overset{d}{\longrightarrow}}%
 
\global\long\def\asto{\overset{a.s.}{\longrightarrow}}%
\global\long\def\tJ{\tilde{{\cal J}}}%

\title{Identification of Semiparametric Panel Multinomial\\
Choice Models with Infinite-Dimensional Fixed Effects\thanks{Researcher(s)\textquoteright{} own analyses calculated (or derived)
based in part on data from Nielsen Consumer LLC and marketing databases
provided through the NielsenIQ Datasets at the Kilts Center for Marketing
Data Center at The University of Chicago Booth School of Business.
The conclusions drawn from the NielsenIQ data are those of the researcher(s)
and do not reflect the views of NielsenIQ. NielsenIQ is not responsible
for, had no role in, and was not involved in analyzing and preparing
the results reported herein.} \thanks{We thank Xiaohong Chen, Peter Phillips, and Phil Haile for their invaluable
advice and encouragement. We thank three anonymous referees for their
comments that significantly improved the paper. We thank Don Andrews,
Isaiah Andrews, Tim Armstrong, Xu Cheng, Tim Christensen, Ben Connault,
Francis Diebold, Bo Honor\'{e}, Joel Horowitz, Yuichi Kitamura, Patrick
Kline, Lixiong Li, Yuan Liao, Charles Manski, Aviv Nevo, Matt Seo,
Xiaoxia Shi, Frank Schorfheide, Elie Tamer, Ed Vytlacil, Rui Wang,
Sheng Xu and participants at various seminars and conferences for
helpful comments. We thank Wenli Lyu and Chuyue Tian for excellent
research assistance.}}
\author{Wayne Yuan Gao\thanks{Gao: Department of Economics, University of Pennsylvania, \protect\href{mailto:waynegao@upenn.edu}{waynegao@upenn.edu}.}
~and Ming Li\thanks{Li: Department of Economics and Risk Management Institute, National
University of Singapore, \protect\href{mailto:mli@nus.edu.sg}{mli@nus.edu.sg}.}}
\maketitle
\begin{abstract}
\noindent This paper proposes a robust method for semiparametric identification
and estimation in panel multinomial choice models, where we allow
for infinite-dimensional fixed effects that enter into consumer utilities
in an additively nonseparable way, thus incorporating rich forms of
unobserved heterogeneity. Our identification strategy exploits multivariate
monotonicity in parametric indices, and uses the logical contraposition
of an intertemporal inequality on choice probabilities to obtain identifying
restrictions. We provide a consistent estimation procedure, and demonstrate
the practical advantages of our method with Monte Carlo simulations
and an empirical illustration on popcorn sales with the NielsenIQ
data.
\end{abstract}
\newpage{}

\setlength{\abovedisplayskip}{3pt} 
\setlength{\belowdisplayskip}{3pt} 

\section{\label{sec:Intro}Introduction}

\noindent This paper proposes a method for semiparametric identification
and estimation in panel multinomial choice models, where we allow
for infinite-dimensional fixed effects that enter into consumer utilities
in an additively nonseparable manner. The proposed method also applies
more widely beyond panel multinomial choice models, and can be adapted
to a wide range of models characterized by \emph{multi-index single-crossing
conditions}, which we introduce later in this paper.

To fix ideas, we start with the following panel multinomial choice
model:
\[
y_{ijt}=\ind\left\{ u\left(X_{ijt}^{'}\b_{0},A_{ij},\e_{ijt}\right)\geq\max_{k\in\left\{ 1,\ldots,J\right\} }u\left(X_{ikt}^{'}\b_{0},A_{ik},\e_{ikt}\right)\right\} ,
\]
where agent $\mathit{i}$'s utility from a candidate product $j$
at time $t$, represented by $u(X_{ijt}^{'}\b_{0},A_{ij},\e_{ijt})$,
is taken to be a function of three components. The first is a linear
index $X_{ijt}^{'}\b_{0}$ of observable characteristics $X_{ijt}$,
which contains a finite-dimensional parameter of interest $\b_{0}$
we will identify and estimate. The second term $A_{ij}$ is an infinite-dimensional
fixed effect that can be heterogeneous across each agent-product combination.
We emphasize that $X_{ijt}$ and $A_{ij}$ can be arbitrarily dependent.
The last term $\e_{ijt}$ is an idiosyncratic time-varying error term
of arbitrary dimension. The three components are then aggregated by
an unknown utility function $u$ in an additively nonseparable way,
with the only restriction being that each agent's utility $u(X_{ijt}^{'}\b_{0},A_{ij},\e_{ijt})$
is\emph{ increasing} in its first argument. Each agent then chooses
a certain product in a given time period, represented by $y_{ijt}=1$,
if and only if this product gives her the highest utility among all
available products.

The infinite dimensionality of the terms $u$, $A_{ij}$, and $\e_{ijt}$,
together with the model\textquoteright s additively non-separable
interaction structure, jointly generates a rich class of unobserved
heterogeneity. Across each agent-product combination $ij$, we are
effectively allowing for nonparametric variations in agent utilities.
Such variation proxies for the effects of complicated unobserved factors
that influence choice behavior, such as brand loyalty, subtle flavors,
and unique styles of products. In addition, we work with a nonparametric
time homogeneity assumption on the error terms $\e_{ijt}$ that restricts
$\e_{ijt}$ and $\e_{ijs}$ from two periods $t$ and $s$ to have
the same marginal distribution given the observed covariates from
the two periods. Apart from this, we impose no parametric restrictions
on the distribution of $\e_{ijt}$ and no additional restrictions
on its dependence across time $t$ and products $j$. In particular,
the fully unrestricted dependence of $\e_{ijt}$ across products $j$
allows our framework to remain robust to the well-known ``Blue-Bus/Red-Bus''
problem and related pathologies that arise in many standard multinomial
choice models, as discussed in \citet*{berry2007pure}.\footnote{See the discussion after Assumption \ref{assu:EpsDist} in Section
\ref{subsec:PMC_Setup} for more details.}

The generality of our framework nests many semiparametric (and parametric)
panel multinomial choice models with scalar fixed effects, scalar
error terms, and varying degrees of additive separability, including
the following standard specification:
\[
y_{ijt}=\ind\left\{ X_{ijt}^{'}\b_{0}+A_{ij}+\e_{ijt}>\max_{k\in\left\{ 1,\ldots,J\right\} }\left(X_{ikt}^{'}\b_{0}+A_{ik}+\e_{ikt}\right)\right\} .
\]
Relative to existing work, our framework accommodates both the infinite
dimensionality of unobserved heterogeneity and non--additive separability
in agent utilities, under a standard time-homogeneity assumption on
the idiosyncratic error term that is widely used in the literature.

Our identification strategy leverages multivariate monotonicity in
contrapositive form. The intuition is straightforward: if the choice
probability of a given product (or subset of products) strictly \textit{increases}
from one period to the next, then it \textit{cannot} be that this
product (or all products in the subset) becomes \textit{worse} while
all other products become \textit{better} over the two periods. By
applying this contraposition to a carefully constructed inequality
in conditional choice probabilities, we obtain an identifying restriction
on the index values that is free of all infinite-dimensional nuisance
parameters. We further show that, in a two-period setting, the identified
set obtained by aggregating these restrictions across all product
subsets is sharp.

Based on our identification result, we provide consistent two-step
set (or point) estimators, together with a computational algorithm
adapted to the technical challenges of our framework. The first stage
takes the form of a standard nonparametric regression, where we estimate
a collection of intertemporal differences in conditional choice probabilities.
In the second stage, we numerically minimize our sample criterion
function with the first-stage estimates plugged in. A highlight of
our computational procedure is the adoption of a spherical-coordinate
reparameterization of our criterion functions in terms of \emph{angles},
which enables us to exploit a combination of topological, geometric
and computational advantages. A simulation study is conducted to analyze
the finite-sample performance of our method and the adequacy of our
computational procedure for practical implementation.

We also provide an empirical illustration of our procedure, where
we use the NielsenIQ data on popcorn sales in the United States to
explore the effects of marketing promotion. The results show that
our procedure produces estimates that conform well with economic intuition.
For example, we find that special in-store displays boost sales not
only through a direct promotion effect but also through the attenuation
of consumer price sensitivity. Intuitively, marketing managers are
more likely to promote products for which they know consumers are
more sensitive to price and promotion. Hence, the average effective
price sensitivities of promoted products tend to be larger than those
not promoted due to the selection effect. Given the non-additive nature
of such selection effects, estimators based on additive separability
will be biased. In contrast, our method is robust to such confounding
effects, thus producing more sensible estimates.

The validity of our identification strategy, as well as our estimation
procedure, relies solely on monotonicity in an index structure, and
thus extends naturally beyond panel multinomial choice models. We
also introduce the \emph{multi-index single-crossing (MISC)} condition
framework, a general econometric framework under which our method
can be applied. This framework encompasses the key ingredients of
a large class of models, such as binary choice models with awareness,
binary choice with endogeneity, dyadic network formation, bilateral
matching, and endogenous censoring.

We acknowledge several limitations of our identification approach.
First, our current model setup does not allow for time-varying endogeneity
between observed covariates $X_{ijt}$ and the error term $\e_{ijt}$,
which effectively rules out the inclusion of contemporaneously endogenous
covariates and/or lagged outcomes. See a subsequent paper by \citet{gao2023identification}
for a weaker version of the time homogeneity assumption that can be
exploited for identification in panel multinomial choice models with
 endogenous and dynamic covariates. Second, our current approach does
not allow for random coefficients on time-varying covariates. That
said, rich forms of time-invariant taste heterogeneity have been absorbed
into the nonparametric fixed effects under our current setting. Third,
in the short-panel setting, our current approach effectively ``differences
out'' the unobserved individual fixed effects $A_{i}$. As a result,
it cannot identify counterfactual parameters that depend on the distribution
of $A_{i}$. However, in long panels, our approach can be adapted
to identify counterfactual parameters, which we discuss in more detail
in Appendix \ref{sec:Counterfactual-Analysis}.

\medskip{}

This paper builds upon and contributes to a large literature on semiparametric
(and parametric) discrete choice models, dating back to \citet*{mcfadden1974conditional}
and \citet*{manski1975maximum}, and more specifically to the line
of literature on panel multinomial choice models. Our work is most
closely related to \citet*{pakes2016moment}, who also exploit weak
monotonicity and time homogeneity, but restrict the effect of unobserved
heterogeneity to be a scalar index that is additively separable from
the index of observable characteristics. \citet*{shi2017estimating}
exploit cyclical monotonicity of \emph{vector}-valued functions in
a fully additive panel multinomial choice model. \citet*{khan2021inference}
consider another additive model, but utilize the subsample of observations
with time-invariant covariates along \emph{all products but one} so
as to leverage univariate monotonicity. \citet*{honore2000panel}
also exploit univariate monotonicity when certain covariates across
two periods are equal in a dynamic panel setting. \citet*{chernozhukov2017nonseparable}
consider a model with an additive effect under an ``on-the-diagonal''
restriction (i.e., when covariates at two different time periods coincide).
By allowing non-additiveness in the specification of utility functions
and infinite-dimensional fixed effects, our method is different from
and thus complementary to those proposed in these aforementioned papers. 

This paper is also connected to the literature on the nonparametric
identification and estimation in discrete choice models (e.g., \citet{berry2014identification,compiani2022market}).
That literature assumes monotonicity restrictions of the demand functions
in product-market specific parametric indices to invert the demand
system. This is particularly useful as it leads to a system of equations
with only one unobservable per equation, from which the unobservable
product-market specific demand shifter can be successfully constructed.
Our paper considers a different model, but also leverages monotonicity
restrictions in parametric indices to facilitate identification and
estimation of the structural parameters. In both cases, the index
assumption restricts the amount of unobserved heterogeneity in preferences
for the variables included in the index.

More broadly, our work connects to the semiparametric literature on
the identification and estimation of models characterized by monotonicity
in a single parametric index. A related class of estimators that leverage
univariate monotonicity, known as \emph{maximum score }or\emph{ rank-order
estimator}s, dates back to a series of important contributions by
\citet*{manski1975maximum,manski1985semiparametric,manski1987semiparametric},
and is further investigated in \citet*{han1987non}, \citet*{horowitz1992smoothed},
\citet*{abrevaya2000rank}, \citet*{honore2002semiparametric}, \citet*{fox2007semiparametric},
and \citet*{yan2019semiparametric}.\footnote{We clarify that \citet{manski1975maximum}, \citet*{fox2007semiparametric}
and \citet*{yan2019semiparametric} consider multinomial choice models
with multiple parametric indices, but focus on settings where the
``multi-index'' problem can be reduced to leverage single-variate
monotonicity (or a single-variate rank-order property).} Despite the similar reliance on monotonicity, the\emph{ multi-index}
nature of our model, and more importantly the multivariate monotonicity
condition that we leverage, induce key differences from the \textit{single-index}
(and univariate monotonicity) setting, leading to a substantially
different estimation method relative to rank-order estimators.

\medskip{}

The rest of this paper is organized as follows. Section \ref{sec:PMC}
introduces our main model specifications and assumptions. Section
\ref{subsec:ID} presents our key identification strategy. In Section
\ref{sec:S_EstComp}, we provide consistent estimators along with
a computational procedure to implement it. Section \ref{sec:Ext_MMIM}
discusses the generalization of our method to the framework of multi-index
single-crossing conditions. Sections \ref{sec:Sim} and \ref{sec:Emp}
switch back to our main panel multinomial choice model, for which
we provide a simulation study and an empirical illustration using
the NielsenIQ data. We conclude in Section \ref{sec:Conclusion}.

\section{\label{sec:PMC}Panel Multinomial Choice Model}

\subsection{\label{subsec:PMC_Setup}Model and Assumptions}

In this section, we present a semiparametric panel multinomial choice
model featuring infinite-dimensional unobserved heterogeneity and
flexible forms of non-separability, which serves as the main framework
for illustrating our identification and estimation method.

Specifically, we consider the following model, in which individual
$i$ chooses product $j$ at time $t$ if and only if $i$ prefers
product $j$ to all other alternatives at time $t$:
\begin{align}
y_{ijt} & =\ind\left\{ u\left(X_{ijt}^{'}\b_{0},A_{ij},\e_{ijt}\right)>\max_{k\in\left\{ 1,\ldots,J\right\} \backslash\left\{ j\right\} }u\left(X_{ikt}^{'}\b_{0},A_{ik},\e_{ikt}\right)\right\} \label{eq:Model_PMC}
\end{align}
where:
\begin{itemize}
\item $i\in\{1,\ldots,N\}$ denotes $N$ \textit{individuals}.
\item $j\in{\cal J}:=\{1,\ldots,J\}$ denotes the set of $J$ choice alternatives,
with \textit{products} indexed by $1,\ldots,J$. Throughout this paper,
we treat the number of products $J$ as fixed.
\item $t\in\{1,\ldots,T\}$ denotes the $T\geq2$ time periods. In this
paper, we consider a short-panel setting in which $T$ is fixed.
\item $X_{ijt}$ is an $\R^{D}$-valued vector of observable characteristics
specific to each agent--product--time tuple $ijt$. These may include,
for example, buyer characteristics such as income, product characteristics
such as price and promotion status, as well as interaction and higher-order
terms of these variables.
\item $y_{ijt}$ is an observable binary variable, with $y_{ijt}=1$ indicating
that buyer $i$ chooses product $j$ at time $t$, and $y_{ijt}=0$
indicating otherwise.
\item $\b_{0}\in\R^{D}$ is the finite-dimensional parameter of interest.
\item $A_{ij}$ represents an $ij$-specific time-invariant unobserved heterogeneity
term of arbitrary dimension, which we refer to as the $ij$-specific
\textit{fixed effect}.
\item $\e_{ijt}$ is an $ijt$-specific unobserved error term of arbitrary
dimension, which captures time-idiosyncratic utility shocks to product
$j$ for agent $i$ at time $t$.
\item $u$ is an unknown function, interpreted as a \textit{utility function}
that aggregates the parametric index $X_{ijt}^{'}\b_{0}$, the fixed
effect $A_{ij}$, and the error term $\e_{ijt}$ into a scalar representing
agent $i$'s utility from choosing product $j$ at time $t$.
\end{itemize}
We first present our main modeling assumptions and then discuss these
assumptions in conjunction with our model specification \eqref{eq:Model_PMC}.
To economize on notation, we refer to the collection of variables
concatenated along product and time dimensions: $\X_{it}=(X_{ijt})_{j=1}^{J}$,
$\X_{i}=(\X_{it})_{t=1}^{T}$, ${\bf A}_{i}=(A_{ij})_{j=1}^{J}$,
$\boldsymbol{\e}_{it}=(\e_{ijt})_{j=1}^{J}$, and $\boldsymbol{\e}_{i}=(\boldsymbol{\e}_{it})_{t=1}^{T}$.
We also write $\d_{ijt}=X_{ijt}^{'}\b_{0}$ to denote the parametric
index. 
\begin{assumption}[Cross-Sectional Random Sampling]
\label{assu:RandSamp} $({\bf Y}_{i},\X_{i},{\bf A}_{i},\boldsymbol{\epsilon}_{i})$
is i.i.d. across $i\in\left\{ 1,\ldots,N\right\} $ with $N\to\infty$.
\end{assumption}
\noindent Assumption \ref{assu:RandSamp} is a standard assumption.\footnote{It is worth noting that we have not yet imposed any explicit restrictions
on the structure of the spaces in which the arbitrary-dimensional
random elements ${\bf A}_{i}$ and $\bs{\e}_{i}$ are defined. However,
implicit in our model specification and in Assumption \ref{assu:RandSamp}
is the requirement that $({\bf Y}_{i},\X_{i},{\bf A}_{i},{\bf \e}_{i})$
be well-defined as random elements---that is, measurable functions---on
a sufficiently rich probability space $(\O,\mathscr{F},\P)$.} Recall that the number of time periods $T$ is held fixed, and we
focus on a short panel setting with cross-sectional asymptotics.
\begin{assumption}[Monotonicity in the Index]
\label{assu:PMC_Mono} For every realization of $(A_{ij},\e_{ijt})$,
the mapping $\tilde{\d}\longmapsto u(\tilde{\d},A_{ij},\e_{ijt})$
is weakly increasing in the scalar-valued argument $\tilde{\d}$.
\end{assumption}
\noindent Essentially, Assumption \ref{assu:PMC_Mono} states that
$u(\d_{ijt},A_{ij},\e_{ijt})$, the utility of individual $i$ from
choosing product $j$ at time $t$, is weakly increasing\footnote{It should be clarified that increasingness is without loss of generality
given monotonicity, since the index $\d_{ijt}=X_{ijt}^{'}\b_{0}$
contains an unknown parameter with unrestricted signs.} in the index $\d_{ijt}$. Given the index structure, monotonicity
itself is a relatively mild assumption. In the standard panel multinomial
choice model with scalar-valued $A_{ij}$ and $\e_{ij}$ along with
additive $u$, i.e., $u(\d_{ijt},A_{ij},\e_{ijt})=\d_{ijt}+A_{ij}+\e_{ijt}$,
Assumption \ref{assu:PMC_Mono} is trivially satisfied (with strictness)
by construction.

In a way, Assumption \ref{assu:PMC_Mono} endows the index $\d_{ijt}$
and the parameter $\b_{0}$ with economic interpretations. Under Assumption
\ref{assu:PMC_Mono}, $\d_{ijt}$ may be considered as a quality measure
of the match between agent $i$ and product $j$ based on their observable
characteristics at time $t$, inducing an interpretation of $\b_{0}$
as representing how a certain change in a linear combination of observable
characteristics may increase utilities for \textit{all} agents from
a certain product $j$, \emph{ceteris paribus}. Hence, without Assumption
\ref{assu:PMC_Mono}, it would be hard to interpret $\b_{0}$. 

Moreover, the monotonicity restriction is imposed on $\d_{ijt}$,
but not directly on any specific observable characteristics in $X_{ijt}$:
quadratic or higher-order polynomial terms, and other functions of
observable characteristics can be included in $X_{ijt}$ whenever
appropriate.
\begin{assumption}[Pairwise Time Homogeneity]
\label{assu:EpsDist}The distributions of $\boldsymbol{\e}_{it}$
and $\boldsymbol{\e}_{is}$ conditional on $(\X_{it},\X_{is},{\bf A}_{i})$
across any pair of periods $t\neq s\in\{1,\ldots,T\}$ satisfy
\[
\rest{\boldsymbol{\e}_{it}}\left(\X_{it},\X_{is},{\bf A}_{i}\right)\rest{\sim\boldsymbol{\e}_{is}}\left(\X_{it},\X_{is},{\bf A}_{i}\right).
\]
\end{assumption}
\noindent Assumption \ref{assu:EpsDist}, a multinomial extension
of the group homogeneity assumption in \citet*{manski1987semiparametric},
is also a standard assumption in the literature on panel multinomial
choice models, such as in \citet*{chernozhukov2017nonseparable},
\citet*{shi2017estimating}, and \citet*{pakes2016moment}.\footnote{In particular, \citet*{pakes2016moment} investigate the following
panel multinomial choice model:
\begin{equation}
y_{ijt}=\ind\left\{ g_{j}\left(X_{ijt},\b_{0}\right)+f_{j}\left(A_{ij},\e_{ijt}\right)>\max_{k\neq j}g_{k}\left(X_{ikt},\b_{0}\right)+f_{k}\left(A_{ik},\e_{ikt}\right)\right\} ,\label{eq:Model_PakesPorter}
\end{equation}
where the function $g_{j}$ produces a potentially nonlinear parametric
index and $f_{j}$ aggregates fixed effects and idiosyncratic errors
into a scalar value in a nonseparable way, while additive separability
between the observable covariate index $g_{j}(X_{ijt},\b_{0})$ and
the unobserved heterogeneity index $f_{j}(A_{ij},\e_{ijt})$ is still
maintained. Moreover, although the dimensions of $A_{ij}$ and $\e_{ijt}$
are not restricted in \citet*{pakes2016moment}, their joint effect
is effectively summarized by a single scalar index $f_{j}(A_{ij},\e_{ijt})$.
We reiterate that our model \eqref{eq:Model_PMC} not only incorporates
infinite-dimensionality in unobserved heterogeneity as captured by
$A_{ij}$ and $\e_{ijt}$, but also allows such heterogeneity to enter
into agent utility functions in a fully \textit{nonseparable }way.} As shown in the next subsection, Assumption \ref{assu:EpsDist} is
the key condition underlying our identification strategy. Essentially,
we rely on Assumption \ref{assu:EpsDist} to link (a particular class
of) intertemporal changes in conditional choice probabilities between
two periods $s$ and $t$ to intertemporal changes in the parametric
indices of all products, $(\d_{ijt}-\d_{ijs})_{j\in{\cal J}}$, thereby
yielding identifying restrictions for the parameter $\b_{0}$. Note
that Assumption \ref{assu:EpsDist} concerns only the marginal distributions
of $\boldsymbol{\e}_{it}$ in different periods, and we make no explicit
assumptions about the serial dependence between $\boldsymbol{\e}_{is}$
and $\boldsymbol{\e}_{it}$, nor do we impose any restrictions on
the dependence structure of $A_{ij}$ and $\e_{ijt}$ across products
$j\in{\cal J}$.

It is worth emphasizing that the absence of restrictions on the cross-product
dependence structure of $\e_{ijt}$ in our setup also enables our
choice model to circumvent the well-known ``Blue-Bus/Red-Bus'' problem\footnote{See, for example, \citet*{mcfadden1974conditional} and \citet{train2009discrete}
for descriptions of the Independence of Irrelevant Alternatives property
and the ``Blue-Bus/Red-Bus'' problem.} that arises in discrete choice models with additive errors that are
assumed to be independent across products. Specifically, consider
a standard multinomial logit model: if we arbitrarily create a new
dummy product by replicating an existing one and giving it a different
name, then under various (mixed) logit models a new independent copy
of logit error is drawn for this new dummy product, which results
in a strict increase in the consumers' indirect utilities, even though
the new product is a simple duplicate of an existing product. This
``Blue-Bus/Red-Bus'' phenomenon also implies that consumers' indirect
utilities would diverge to infinity if we keep adding such dummy new
products, and that consumers will choose one of these duplicate products
with increasingly higher probabilities, both of which are unrealistic.
This problem, as well as other conceptual issues with independent
additive errors, has been well discussed, say, in \citet{berry2007pure}.
We emphasize that our current model does \emph{not} lead to the ``Blue-Bus/Red-Bus''
problem, since we allow errors to be arbitrarily correlated across
products. Hence, when duplicate products are created, the errors of
duplicate products are allowed to be perfectly correlated (as they
should be), so that adding duplicate products with different name
labels does not blow up the indirect utility of the consumer, nor
does it make the consumer more likely to choose one of the duplicates
relative to any set of remaining products.

That said, Assumption \ref{assu:EpsDist} does entail certain limitations
for the model. First, the assumption effectively rules out time-varying
endogeneity\footnote{Note that time-invariant endogeneity can be incorporated through the
unrestricted dependence between $\e_{ijt}$ and the fixed effect $A_{ij}$,
which can be arbitrarily correlated with ${\bf X}_{i}$ (in time-invariant
manners).}: for example, if ${\bf X}_{it}\neq{\bf X}_{is}$, the conditional
distribution of $\boldsymbol{\e}_{it}$ is nevertheless assumed to
be the same as that of $\boldsymbol{\e}_{is}$, which is likely to
be violated if the distribution of $\boldsymbol{\e}_{it}$ covaries
with that of ${\bf X}_{it}$. This is admittedly a limitation of the
approach, though it is common to this line of literature that utilizes
Assumption \ref{assu:EpsDist} as cited before. Nonetheless, subsequent
work by \citet*{gao2023identification} has proposed a weakened notion
of the stationarity/homogeneity assumption that is only imposed on
``exogenous covariates'', and has proposed an approach for partial
identification, albeit in a slightly different setting with additive
scalar-valued fixed effects. Relatedly, \citet{li2024identificationestimationtimevaryingendogenous}
proposes a correlated random coefficient linear panel model in which
regressors can be correlated with time-varying, individual-specific
random coefficients, thereby accommodating time-varying endogeneity
in the covariates. Second, a further limitation of Assumption \ref{assu:EpsDist}
is that it rules out random coefficients, a modeling device that was
popularized by \citet*{blp1995io} due to its ability to generate
rich substitution patterns among products with multi-dimensional observable
characteristics. However, the flexibility afforded by our general
fixed effect specification can incorporate arbitrarily complicated
substitution patterns with respect to \emph{time-invariant} components
of observed and unobserved product characteristics. Our infinite-dimensional
fixed-effect approach is thus more suitable to panel-data settings
where researchers are more interested in incorporating an arbitrarily
complicated form of time-invariant heterogeneity across agent-product
pairs. 

Beyond Assumption \ref{assu:EpsDist}, another limitation of the paper
is that we treat the distribution of unobserved heterogeneity (i.e.,
the fixed effects) as a nuisance parameter and focus on the identification
of $\b_{0}$. Many counterfactual parameters require knowledge of
the distribution of the relevant unobserved heterogeneity terms, which
is not identified given that the fixed effects are ``differenced out''
under our current approach in a short-panel setting. To this end,
we discuss in Appendix \ref{sec:Counterfactual-Analysis} how to use
the estimated $\b_{0}$ in a long panel setting ($T\to\infty$) to
perform counterfactual analysis. The idea is when a long panel is
available, we can consistently estimate for each individual certain
parameters of interest by using the observations only from that individual,
which effectively controls for the unobserved ${\bf A}_{i}$.

Furthermore, there are interesting economic questions that can be
addressed based on the knowledge of $\b_{0}$, and we discuss two
of them here.\footnote{We thank an anonymous referee for the suggestions here.}
First, consider research questions that focus on the existence and
direction of an effect, which can often be determined by the relative
magnitudes of effects from covariates as captured by $\b_{0}$. For
instance, one may ask whether advertising by a competitor within the
same product category exerts a positive or negative influence on the
demand for a focal brand. Theoretical considerations offer plausible
arguments for both substitution and category-expansion effects \citep*{simon1980shape,narayanan2004return},
rendering the sign of the net impact an empirical matter of interest.
Second, the proposed approach may also be applied to model specification
testing (or as a robustness check) for models with more restrictive
specifications on the fixed effects and errors. Specifically, by comparing
the estimate of $\b_{0}$ obtained from our model with that from a
more standard model (say, multinomial logit with fixed effects), one
can assess whether the parametric restrictions, particularly those
pertaining to unobserved heterogeneity, are supported by the data.
This aligns with the broader econometric literature on specification
testing (e.g., \citet*{hausman1978specification,vuong1989likelihood}),
and can be particularly useful for evaluating the validity of assumptions
regarding the distributional structure or functional form of heterogeneity.\footnote{\citet*{otaspecification} develop a specification test of parametric
binary choice models via the maximum score estimator. Given that our
approach generalizes the maximum score idea to settings with multivariate
monotonicity and infinite-dimensional fixed effects, extending the
specification testing idea of \citet*{otaspecification} to our multivariate
monotone panel setting can be a promising avenue for future work.}

\subsection{\label{subsec:ID}Key Identification Strategy}

In this section, we present our main semiparametric identification
result for model \eqref{eq:Model_PMC} under Assumptions \ref{assu:PMC_Mono}
and \ref{assu:EpsDist}, and detail our key identification strategy,
which exploits multivariate monotonicity in the presence of additive
non-separability and nonparametric fixed effects.

To start, fix any subset of products $\tJ\subseteq{\cal J}$, a pair
of time periods $t\neq s\in\{1,\ldots,T\}$ and a generic realization
of observable covariates in the two periods $t$ and $s$, i.e., $({\bf X}_{it},{\bf X}_{is})=({\bf x}_{t},{\bf x}_{s})\in\text{Supp}({\bf X}_{it},{\bf X}_{is})$.
For simpler notation, we write ${\bf X}_{i,ts}=({\bf X}_{it},{\bf X}_{is})$,
${\bf x}_{ts}=({\bf x}_{t},{\bf x}_{s})$, $\d_{jt}=x_{jt}^{'}\b_{0}$,
where $x_{jt}$ denotes the $j$-th column of ${\bf x}_{t}$.

For each individual $i$, consider the intertemporal change in the
probability of that individual choosing any product $j\in\tJ$ across
periods $t$ and $s$, conditional on $({\bf X}_{i,ts},{\bf A}_{i})$,
the joint realization of the fixed effect and the observable covariates
in both periods $t$ and $s$. Formally, writing $y_{i\tJ t}=\sum_{j\in\tJ}y_{ijt}$,
we have
\begin{align}
 & \E\left[\rest{y_{i\tJ t}-y_{i\tJ s}}{\bf X}_{i,ts}={\bf x}_{ts},{\bf A}_{i}\right]\nonumber \\
= & \int\ind\left\{ \max_{j\in\tJ}u\left(\d_{jt},A_{ij},\e_{ij{\color{red}t}}\right)>\max_{k\notin\tJ}u\left(\d_{kt},\,A_{ik},\e_{ik{\color{red}t}}\right)\right\} \text{d}{\cal \P}\left(\rest{\bs{\e}_{i{\color{red}t}}}{\bf X}_{i,ts}={\bf x}_{ts},{\bf A}_{i}\right)\nonumber \\
 & -\int\ind\left\{ \max_{j\in\tJ}u\left(\d_{js},A_{ij},\e_{ij{\color{red}s}}\right)>\max_{k\notin\tJ}u\left(\d_{ks},\,A_{ik},\e_{ik{\color{red}s}}\right)\right\} \text{d}{\cal \P}\left(\rest{\bs{\e}_{i{\color{red}s}}}{\bf X}_{i,ts}={\bf x}_{ts},{\bf A}_{i}\right)\nonumber \\
= & \int\left[\begin{array}{c}
\ind\left\{ \max_{j\in\tJ}u\left(\d_{jt},A_{ij},\tilde{\e}_{j}\right)>\max_{k\notin\tJ}u\left(\d_{kt},A_{ik},\tilde{\e}_{k}\right)\right\} \\
-\ind\left\{ \max_{j\in\tJ}u\left(\d_{js},A_{ij},\tilde{\e}_{j}\right)>\max_{k\notin\tJ}u\left(\d_{ks},A_{ik},\tilde{\e}_{k}\right)\right\} 
\end{array}\right]\text{d}{\cal \P}\left(\rest{\tilde{\bs{\e}}}{\bf X}_{i,ts}={\bf x}_{ts},{\bf A}_{i}\right),\label{eq:dyXA}
\end{align}
where the last equality follows from Assumption \ref{assu:EpsDist}
(pairwise time homogeneity):
\[
\rest{\bs{\e}_{i{\color{red}s}}\sim\bs{\e}_{i{\color{red}t}}\sim\tilde{\bs{\e}}}{\bf X}_{i,ts}={\bf x}_{ts},{\bf A}_{i}
\]
in which $\tilde{\bs{\e}}$ is a name-holder random element with the
shared marginal distribution of $\bs{\e}_{i{\color{red}t}}$ and $\bs{\e}_{i{\color{red}s}}$
given $({\bf x}_{ts},{\bf A}_{i})$.

Since $u$ is weakly increasing in its index argument, the choice
indicator
\[
\ind\left\{ \max_{j\in\tJ}u\left(\d_{jt},A_{ij},\tilde{\e}_{j}\right)>\max_{k\notin\tJ}u\left(\d_{kt},A_{ik},\tilde{\e}_{k}\right)\right\} 
\]
is weakly increasing in the vector $(\d_{jt})_{j\in{\cal \tJ}}$ and
decreasing in the remaining vector $(\d_{kt})_{k\notin{\cal \tilde{J}}}$
for every possible realization of $\tilde{\bs{\e}}$ and ${\bf A}_{i}$.
Hence, if 
\begin{equation}
\d_{jt}\leq\d_{js},\ \forall j\in\tJ\quad\text{and}\quad\d_{kt}\geq\d_{ks}\ \forall k\notin\tJ,\label{eq:dXb_jk}
\end{equation}
then, for every possible realization of $\tilde{\bs{\e}}$ and ${\bf A}_{i}$,
we have
\begin{equation}
\left[\begin{array}{c}
\ind\left\{ \max_{j\in\tJ}u\left(\d_{jt},A_{ij},\tilde{\e}_{j}\right)>\max_{k\notin\tJ}u\left(\d_{kt},A_{ik},\tilde{\e}_{k}\right)\right\} \\
-\ind\left\{ \max_{j\in\tJ}u\left(\d_{js},A_{ij},\tilde{\e}_{j}\right)>\max_{k\notin\tJ}u\left(\d_{ks},A_{ik},\tilde{\e}_{k}\right)\right\} 
\end{array}\right]\leq0,\label{eq:dyXAe>0}
\end{equation}
and consequently
\begin{equation}
\E\left[\rest{y_{i\tJ t}-y_{i\tJ s}}{\bf X}_{i,ts}={\bf x}_{ts},{\bf A}_{i}\right]\leq0\quad\text{for every possible realization of }{\bf A}_{i}.\label{eq:dyXA>0}
\end{equation}

Now, consider the observable intertemporal change in conditional choice
probabilities:
\begin{align}
\g_{\tJ,ts}\left({\bf x}_{ts}\right) & :=\E\left[\rest{y_{i\tJ t}-y_{i\tJ s}}{\bf X}_{i,ts}={\bf x}_{ts}\right].\label{eq:dyX}
\end{align}
Then, whenever \eqref{eq:dXb_jk} holds, by \eqref{eq:dyXA>0} we
can deduce that
\begin{align*}
\g_{\tJ,ts}\left({\bf x}_{ts}\right) & =\int\underset{\leq0}{\underbrace{\E\left[\rest{y_{i\tJ t}-y_{i\tJ s}}{\bf X}_{i,ts}={\bf x}_{ts},{\bf A}_{i}\right]}}\text{d}\P\left(\rest{{\bf A}_{i}}{\bf X}_{i,ts}={\bf x}_{ts}\right)\leq0.
\end{align*}
In summary, since \eqref{eq:dXb_jk} implies inequality \eqref{eq:dyXAe>0}
for every possible realization of $\tilde{\bs{\e}}$ and ${\bf A}_{i}$,
this inequality will be preserved after $\tilde{\bs{\e}}$ and ${\bf A}_{i}$
are integrated out \emph{cross-sectionally} with respect to the conditional
distribution $\P\left(\rest{\tilde{\bs{\e}},{\bf A}_{i}}{\bf X}_{i,ts}={\bf x}_{ts}\right)$,
regardless of how complicated this unknown conditional distribution
may be. 

The next proposition formalizes the identification strategy described
above, which produces an identifying restriction on the parameter
$\b_{0}.$
\begin{prop}[Key Identifying Restrictions]
\label{prop:ID_rest} Under model \eqref{eq:Model_PMC} and Assumptions
\ref{assu:RandSamp}--\ref{assu:EpsDist}, 
\begin{equation}
\g_{\tJ,ts}\left({\bf x}_{ts}\right)>0\ \text{\ensuremath{\imp}}\ \textup{NOT}\ \left\{ \left(x_{jt}-x_{js}\right)^{'}\b_{0}\leq0,\ \forall j\in\tJ\ \mathrm{and}\ \left(x_{kt}-x_{ks}\right)^{'}\b_{0}\geq0\ \forall k\notin\tJ\right\} \label{eq:ID_rest}
\end{equation}
for any (ordered) pair of time periods $(t,s)$ with $t\neq s\in\{1,\ldots,T\}$,
any subset of products $\tJ\subseteq{\cal J},$ and any realization
of observables ${\bf x}_{ts}\in\text{Supp}(\X_{i,ts})$.
\end{prop}
\noindent Proposition \ref{prop:ID_rest} establishes an identifying
restriction on $\b_{0}$ that is free of all unknown nonparametric
heterogeneity terms---$u$, ${\bf A}$, and $\bm{\e}$---and holds
in the presence of additive non-separability and nonparametric fixed
effects. Proposition \ref{prop:ID_rest} is also intuitive: if the
total market share of the products in $\tJ$ increases between two
periods, then it cannot be the case that the indices of products in
$\tJ$ have all (weakly) worsened while those of products not in $\tJ$
have all (weakly) improved.
\begin{thm}[Identified Set]
\label{thm:SetID} Let $B_{0}$ be the set of all $\b\in\R^{D}$
such that \eqref{eq:ID_rest} holds with $\b$ in lieu of $\b_{0}$,
for almost all ${\bf x}_{ts}\in\text{Supp}(\X_{i,ts})$, all $t\neq s\in\{1,\ldots,T\}$,
and all $\tJ\subseteq{\cal J}$. Then, under model \eqref{eq:Model_PMC}
and Assumptions \ref{assu:RandSamp}--\ref{assu:EpsDist}, $\b_{0}\in B_{0}.$
\end{thm}
\noindent We refer to $B_{0}$ as the \emph{identified set}. In Appendix
\ref{sec:App_PID}, we provide sufficient conditions for point identification
of $\b_{0}$ up to a scale normalization. The assumptions we impose
are similar to those used for point identification in the maximum-score
literature, such as \citet*{manski1985semiparametric}, and in related
work on panel multinomial choice models, such as \citet*{shi2017estimating}
and \citet*{khan2021inference}.

Relative to the well-known maximum-score criterion function studied
by \citet{manski1985semiparametric,manski1987semiparametric} under
univariate monotonicity, our criterion function is non-standard. This
non-standardness arises from a key distinction between multivariate
and univariate monotonicity. To see this more clearly, consider the
special case of a \emph{single-index} setting $(J=1)$\footnote{This arises naturally in binomial choice models with the characteristics
of the outside option set to be zero. In this case, even though there
are nominally two choice alternatives, choice behavior is completely
determined by a single index based on the characteristics of the non-default
option.}, in which case the following equivalence relationship holds given
the \emph{univariate }monotonicity in the index:
\begin{equation}
\left\{ \g\left({\bf x}_{ts}\right)>0\right\} \ \iff\ \left\{ \left(x_{t}-x_{s}\right)^{'}\b>0\right\} ,\label{eq:Equiv_MaxScore}
\end{equation}
Such an ``if-and-only-if'' relationship is a unique feature of the
single-index setting that \emph{cannot} be generalized to the multi-index
setting with $J\geq2$, as the right-hand side of \eqref{eq:ID_rest},
\[
\mathrm{NOT}\ \left\{ \left(x_{jt}-x_{js}\right)^{'}\b_{0}\leq0,\ \forall j\in\tJ\ \mathrm{and}\ \left(x_{kt}-x_{ks}\right)^{'}\b_{0}\geq0\ \forall k\notin\tJ\right\} ,
\]
\emph{does not} imply $\g_{\tJ,ts}\left({\bf x}_{ts}\right)\geq0$
in the converse direction. This breaks the ``if-and-only-if'' relationship
that the maximum-score criterion function in \citet{manski1985semiparametric,manski1987semiparametric}
is built upon. Thus, the maximum-score estimator does not generalize
to multi-index settings. The lack of ``if-and-only-if'' relationship
in the multi-index setting leads to a key difference in the criterion
functions, and consequently a different estimation approach. Importantly,
while the original maximum score criterion and estimator cannot be
generalized to multi-index settings, our procedure can be applied
under a general econometric framework characterized by \emph{multi-index
single-crossing} conditions, which we introduce in Section \ref{sec:Ext_MMIM}.

\subsection{\label{subsec:Sharp}Two-Period Sharpness}

We now establish the sharpness of the identified set $B_{0}$ in a
two-period setting. This sharpness result can be interpreted as the\emph{
pairwise sharpness} of the identifying restrictions in \eqref{eq:ID_rest}:
for fixed periods $s$ and $t$ such that $s<t$, the inequality restrictions
in \eqref{eq:ID_rest} for $(s,t)$ and $(t,s)$ exhaust all the identifying
information available from the model (its specification and assumptions)
and from the distribution of the observable data in periods $s$ and
$t$.
\begin{thm}[Pairwise Sharpness]
\label{thm:Id_sharp} Under model \eqref{eq:Model_PMC} and Assumptions
\ref{assu:RandSamp}--\ref{assu:EpsDist}, $B_{0}$ is sharp for
$T=2$.
\end{thm}
\noindent The proof of Theorem \ref{thm:Id_sharp} exploits and generalizes
a corresponding result in \citet*{pakes2016moment}. Specifically,
\citet*{pakes2016moment} consider a specification where the utility
index $X_{ijt}^{'}\b_{0}$ and an unobserved heterogeneity index $\l\left(A_{ijt},\e_{ijt}\right)$
are additively separable, and establish the sharpness of their identification
result by showing the existence of a nonnegative solution to a system
of linear equations. Here, we consider a more general setup without
requiring additive separability and propose a correspondingly more
general identification argument. Nevertheless, we show that the sharpness
of our identification result under our more general setup can be reduced
to the nonnegative solvability of the same system of linear equations
in \citet*{pakes2016moment}. Hence, by the result in \citet*{pakes2016moment},
our identification result is sharp.

Admittedly, Theorem \ref{thm:Id_sharp} only establishes sharpness
in a two-period setting; however, it does not directly imply ``all-period''
sharpness for $T\geq3$. While the existence of an observationally
equivalent latent error distribution can be established for any realization
${\bf X}_{i,ts}$ and any pair of periods $\left(t,s\right)$ as in
the proof of Theorem \ref{thm:Id_sharp}, to establish the stronger
``all-period'' sharpness result, we would need to show in addition
that there exists an all-period joint distribution of latent errors
that matches all-period observed joint conditional choice probabilities
and satisfies the pairwise time homogeneity assumption. This appears
to be a technically cumbersome exercise, given that the pairwise time-homogeneity
assumption enters as an implicit aggregate restriction on the two-period
error distributions (with all other periods aggregated out). We thus
do not pursue ``all-period'' sharpness here, and only present the
sharpness result above as in \citet*{pakes2016moment}, which also
focuses on two-period (pairwise) sharpness.

\section{\label{sec:S_EstComp}Estimation and Computation}

\subsection{\label{subsec:Criterion}Formulation of Population Criterion Function}

We now propose a population criterion function that encodes the identifying
information in Proposition \ref{prop:ID_rest}. We represent the right-hand
side of \eqref{eq:ID_rest} in Boolean algebra by
\begin{align}
\l_{\tJ}\left({\bf x}_{ts};\b\right) & :=\prod_{k=1}^{J}\ind\left\{ \left(-1\right)^{\ind\left\{ k\in\tJ\right\} }\left(x_{kt}-x_{ks}\right)^{'}{\color{blue}\b}\geq0\right\} ,\label{eq:lambda}
\end{align}
where $\left(-1\right)^{\ind\left\{ k\in\tJ\right\} }$ takes the
value $-1$ for $k\in\tJ$ and $1$ for $k\notin\tJ$. Therefore,
Proposition \ref{prop:ID_rest} can be written algebraically as: $\g_{\tJ,t,s}\left({\bf x}_{ts}\right)>0$
implies $\l_{\tJ}\left({\bf x}_{ts};\b_{0}\right)=0$ for any ${\bf x}_{ts}\in\text{Supp}\left(\X_{i,ts}\right)$. 

We now define the following criterion function by taking a cross-sectional
expectation over the random realization of $\X_{i,ts}$ and aggregating
over all subsets $\tJ\subseteq{\cal J}$:
\begin{align}
Q_{t,s}\left(\b\right) & :=\sum_{\tJ\subseteq{\cal J}}\E\left[\ind\left\{ \g_{\tJ,t,s}\left(\X_{i,ts}\right)>0\right\} \l_{\tJ}\left(\X_{i,ts};\b\right)\right],\label{eq:Q_ts}
\end{align}
which is nonnegative and minimized to zero at $\b_{0}$. Without normalization
and further assumptions for point identification, there could be multiple
values of $\b$ that minimize $Q_{t,s}$ to zero.

More generally, fix any function $G:\R\to\R$ that is \emph{one-sided
sign preserving}, i.e., $G\left(z\right)>0$ for $z>0$ and $G\left(z\right)=0$
for $z\leq0$. For example, we can choose $G\left(z\right)=\left[z\right]_{+}$
where $\left[z\right]_{+}$ is the positive part function. Then, we
define $Q_{t,s}^{G}$ as
\begin{align}
Q_{t,s}^{G}\left(\b\right) & :=\sum_{\tJ\subseteq{\cal J}}\E\left[G\left(\g_{\tJ,t,s}\left(\X_{i,ts}\right)\right)\l_{\tJ}\left(\X_{i,ts};\b\right)\right],\label{eq:Q_G_jts}
\end{align}
which is also minimized to zero at $\b_{0}$. The sign-preserving
function $G$, if further set to be monotone, continuous, or bounded,
serves as a \emph{smoothing} function that can improve the finite-sample
performance of our estimators. We provide more discussions on function
$G$ in the next section, when we construct estimators based on the
sample analog of the population criterion function defined here. 

$Q_{t,s}^{G}$ above is defined for a fixed pair of periods $\left(t,s\right)$,
but in practice we may utilize the information across all pairs of
periods by defining the aggregated criterion function:
\begin{equation}
Q^{G}\left(\b\right):=\sum_{t\neq s}^{T}Q_{t,s}^{G}\left(\b\right),\quad\text{for any }\b\in\R^{D}.\label{eq:PMC_BigQ}
\end{equation}
For notational simplicity, we suppress $G$ in $Q_{t,s}^{G}$ and
$Q^{G}$ in the rest of this paper.

\subsection{\label{subsec:Est}Two-Step Semiparametric Estimation}

We construct our estimator as a semiparametric two-step M-estimator
based on \eqref{eq:PMC_BigQ}. The first stage of our procedure is
concerned with nonparametrically estimating the intertemporal differences
in conditional choice probabilities of the following form:
\[
\g_{\tJ,t,s}\left({\bf x}_{ts}\right)=\sum_{j\in\tJ}\g_{j,t,s}\left({\bf x}_{ts}\right),
\]
where $\g_{j,t,s}\left({\bf x}_{ts}\right)=\E\left[\rest{y_{ijt}-y_{ijs}}\X_{i,ts}={\bf x}_{ts}\right]$
can be separately estimated for each product $j\in{\cal J}.$\footnote{In practice, we only need to estimate $\g_{j,t,s}$ for $\left(J-1\right)$
products and $\frac{1}{2}T\left(T-1\right)$ \emph{ordered} pairs
of periods. The former is because conditional choice probabilities
must sum to one across all $J$ products. Hence, the estimator for
the last product from the other $\left(J-1\right)$ estimates can
be directly derived by $\g_{J,t,s}=-\sum_{j=1}^{J-1}\g_{j,t,s}$.
The latter is because $\g_{j,t,s}=-\g_{j,s,t}$ by construction, so
we may estimate it for either $\left(t,s\right)$ or $\left(s,t\right)$
pair. Notice, however, that each ordered pair $\left(t,s\right)$
or $\left(s,t\right)$ provides complementary identifying information,
as $\l\left(\X_{i,ts};\b\right)$ and $\l\left(\X_{i,st};\b\right)$
do not admit such kind of deterministic relationships.} We note that the first stage estimation includes the observable characteristics
of all products $J$. For example, when $J=3$ and $D=3$, there are
$3\times3\times2=18$ variables in the conditioned set of $\gamma$.
Given the potentially large number of regressors, one may want to
use neural networks \citep{bach2017breaking,chen1999improved} or
penalized sieves \citep{chen_2013} for the first-step estimation.

Given the first-stage estimators $\hat{\g}_{j,t,s}$ and the smoothing
function $G$, in the second stage we numerically compute minimizers
of the sample criterion function,
\begin{align*}
\hat{Q}\left(\b\right):=\sum_{t\neq s}^{T}\hat{Q}_{\tJ,t,s}\left(\b\right), & \text{ where }\hat{Q}_{t,s}\left(\b\right):=\frac{1}{N}\sum_{i=1}^{N}\sum_{\tJ\subseteq{\cal J}}G\left(\hat{\g}_{\tJ,t,s}\left(\X_{i,ts}\right)\right)\l_{\tJ}\left(\X_{i,ts};\b\right).
\end{align*}

It is worth noting that while $\hat{Q}_{t,s}(\b)$ is defined as a
summation over all $2^{J}$ possible subsets $\tJ\subseteq{\cal J}$,
computationally there is no need to fully evaluate $Q_{\tJ,t,s}\left(\b\right)$
for each possible $\tJ\subseteq{\cal J}$ under a given $\b$ when
the knife-edge cases of $(X_{ijt}-X_{ijs})^{'}\b=0$ are ignorable.\footnote{There are at least two reasons why the ``knife-edge'' events of
the form $(X_{ijt}-X_{ijs})^{'}\b=0$ should be ignored. First, $(X_{ijt}-X_{ijs})^{'}\b=0$
technically occurs with probability zero for all $\b$ provided that
$X_{ijt}\neq X_{ijs}$ almost surely, which is a natural assumption
given that we require $X_{ijt}$ be time-varying. Second, for programming
reasons, it is often practically necessary to ignore knife-edge strict
equalities of continuously valued variables, since such equalities
are extremely sensitive to unavoidable numerical errors induced by
the ``machine epsilon.''} This is because, as long as $\l_{\tJ}\left(\X_{i,ts};\b\right)=0$,
the contribution from the $\tJ$-summand would be zero. However, a
careful inspection of $\l_{\tJ}$ reveals that $\l_{\tJ}\left(\X_{i,ts};\b\right)=1$
only if $\tilde{{\cal J}}=\left\{ j\in{\cal J}:\left(X_{ijt}-X_{ijs}\right)^{'}\b\leq0\right\} $.
Hence, in practical implementation we may simply compute 
\[
\hat{Q}_{t,s}\left(\b\right):=\frac{1}{N}\sum_{i=1}^{N}G\left(\sum_{j\in{\cal J}}\hat{\g}_{j,t,s}\left(\X_{i,ts}\right)\ind\left\{ \left(X_{ijt}-X_{ijs}\right)^{'}\b\leq0\right\} \right).
\]

In addition, the scale of $\b_{0}$ is not identified since $\l_{j}\left(\X_{i,ts};\b\right)$
consists of indicator functions of the form $\ind\left\{ \left(X_{ijt}-X_{ijs}\right)^{'}\b\geq0\right\} $.
Hence, we impose the scale normalization $\b_{0}\in\S^{D-1}:=\left\{ v\in\R^{D}:\norm v=1\right\} $.
Following \citet*{chernozhukov2007estimation}, we define the set
estimator by
\begin{equation}
\hat{B}_{\hat{c}}:=\left\{ \b\in\S^{D-1}:\ \hat{Q}\left(\b\right)\leq\min_{\tilde{\b}\in\S^{D-1}}\hat{Q}\left(\tilde{\b}\right)+\hat{c}\right\} \label{eq:Theta_hat_CHT07}
\end{equation}
with $\hat{c}:=O_{p}\left(c_{N}\log N\right)$. 

We now introduce assumptions for establishing the consistency of $\hat{B}_{\hat{c}}$.
\begin{assumption}[First-Stage Estimation]
\label{assu:FS_Conv} For any $\left(j,t,s\right)$ tuple:
\begin{itemize}
\item[(i)]  $\g_{j,t,s}\in\G$, and $\P\left(\hat{\g}_{j,t,s}\in\G\right)\to1$,
with $\G$ being a $\P$-Donsker class of functions in $L_{2}\left({\bf X}\right)$.
\item[(ii)]  $\norm{\hat{\g}_{j,t,s}-\g_{j,t,s}}_{2}:=\sqrt{\int\left(\hat{\g}_{j,t,s}\left(\X_{i,ts}\right)-\g_{j,t,s}\left(\X_{i,ts}\right)\right)^{2}\mathrm{d}\P\left(\X_{i,ts}\right)}=O_{p}\left(c_{N}\right)$
\textup{with $c_{N}\downto0$.}
\end{itemize}
\end{assumption}
\noindent Through Assumption \ref{assu:FS_Conv} we take as given
the large set of theoretical results on nonparametric regression in
the literature. Many kernel-based and sieve-based methods have been
developed, with their properties demonstrated under various sets of
conditions. See \citet*{wasserman2006all} and \citet*{chen2007sieve}
for more comprehensive surveys.
\begin{assumption}[Nice Smoothing Function]
\label{assu:NiceG} The one-sided sign-preserving function $G:\R\to\R_{+}$
is Lipschitz continuous with a finite Lipschitz constant.
\end{assumption}
\noindent Assumption \ref{assu:NiceG} is stronger than necessary
for consistency per se given that our identification result is valid
with any choice of the one-sided sign-preserving function $G$, nevertheless
we take $G$ to be Lipschitz to simplify the proof.

To state the next assumption, we decompose each row (corresponding
to each product) of ${\bf x}_{t}-{\bf x}_{s}$ as the product of its
norm and its \emph{direction}, i.e., ${\bf x}_{jt}-{\bf x}_{js}\equiv r_{j}\left({\bf x}_{t}-{\bf x}_{s}\right)v_{j}\left({\bf x}_{t}-{\bf x}_{s}\right)$,
where $r_{j}\left({\bf x}_{t}-{\bf x}_{s}\right):=\norm{{\bf x}_{jt}-{\bf x}_{js}}$,
and $v_{j}\left({\bf x}_{jt}-{\bf x}_{js}\right):=\left({\bf x}_{jt}-{\bf x}_{js}\right)/\norm{{\bf x}_{jt}-{\bf x}_{js}}$
if ${\bf x}_{jt}\neq{\bf x}_{js}$ while $v_{j}\left({\bf x}_{jt}-{\bf x}_{js}\right):={\bf 0}$
if ${\bf x}_{jt}={\bf x}_{js}$.
\begin{assumption}[Continuous Distribution of Directions]
\label{assu:NoMass} The marginal distribution of $v_{j}\left({\bf X}_{it}-{\bf X}_{is}\right)$
has no mass point except possibly at ${\bf 0}$ and is not supported
on any proper linear subspace of $\R^{D}$ for each $\left(j,t,s\right)$
tuple.
\end{assumption}
\noindent Assumption \ref{assu:NoMass} ensures the continuity of
the population criterion function. We note that Assumption \ref{assu:NoMass}
is mild: it essentially requires that the \emph{directions} of intertemporal
differences in observable characteristics are continuously distributed
on their own supports. In particular, this allows all but one dimensions
of observable characteristics to be discrete.

With the above assumptions imposed, we now establish the consistency
of our set estimator $\hat{B}_{\hat{c}}$, using the results in \citet*{chernozhukov2007estimation}.
\begin{thm}[Consistency]
\label{thm:Consistency}Under Assumptions \ref{assu:RandSamp}--\ref{assu:NoMass},
the set estimator $\hat{B}_{\hat{c}}$ is consistent in Hausdorff
distance: $d_{H}\left(\hat{B}_{\hat{c}},B_{0}\right)=o_{p}\left(1\right)$,
where $d_{H}\left(\hat{B}_{\hat{c}},B_{0}\right)=\max\left\{ \sup_{\b\in\hat{B}_{\hat{c}}}\inf_{\tilde{\b}\in B_{0}}\norm{\b-\tilde{\b}},\sup_{\b\in B_{0}}\inf_{\tilde{\b}\in\hat{B}_{\hat{c}}}\norm{\b-\tilde{\b}}\right\} .$
Furthermore, if $\b_{0}$ is point-identified on $\S^{D-1}$, $\norm{\hat{\b}-\b_{0}}=o_{p}\left(1\right)$
for any $\hat{\b}\in\hat{B}_{\widehat{c}=0}$.
\end{thm}

\subsection{\label{subsubsec:Comp}Computation}

We now explain how we implement the semiparametric two-step estimation
procedure proposed above. Since the model's criterion function is
possibly non-convex, standard gradient-based optimizers are susceptible
to converging to local minima. To address this, we employ a multi-stage
adaptive-grid search algorithm that explores the parameter space more
robustly and aims to locate the global minimizer of the objective
function. The code for our computation algorithm is publicly available
on GitHub.\footnote{https://github.com/mingliecon/GL\_PMC}

\subsubsection*{Choice of the Smoothing Function $G$}

Besides the requirement of Lipschitz continuity in Assumption \ref{assu:NiceG},
in practice we take $G$ to be bounded from above by setting $G\left(z\right)=2\Phi\left(\left[z\right]_{+}\right)-1$,
where $\Phi$ is the standard normal CDF. We now motivate our choice
of $G$.

Recall that our identification strategy is based on the logical implication
of the event $\g_{\tJ,t,s}\left({\bf x}_{ts}\right)>0$. Thus, for
identification purposes we are only interested in $\ind\left\{ \g_{\tJ,t,s}({\bf x}_{ts})>0\right\} $,
i.e., whether the event $\g_{\tJ,t,s}\left({\bf x}_{ts}\right)>0$
occurs, but not in the exact magnitude of $\g_{\tJ,t,s}\left({\bf x}_{ts}\right)$.
However, when $\g_{\tJ,t,s}\left({\bf x}_{ts}\right)$ is close to
zero, the estimator $\hat{\g}_{\tJ,t,s}\left({\bf x}_{ts}\right)$
is relatively more likely to have the wrong sign, so that the plug-in
estimator $\ind\left\{ \hat{\g}_{\tJ,t,s}\left({\bf x}_{ts}\right)>0\right\} $
may induce a large error of magnitude $1$. Hence, the smoothing by
$G$ helps down-weight the observations when $\hat{\g}_{\tJ,t,s}\left({\bf x}_{ts}\right)$
is close to zero and shrinks the magnitude of possible errors.

On the other hand, when $\g_{\tJ,t,s}({\bf x}_{ts})$ is positive
and large so that $\ind\left\{ \g_{\tJ,t,s}({\bf x}_{ts})>0\right\} $
can be estimated well, the magnitude of $\g_{\tJ,t,s}({\bf x}_{ts})$
itself does not provide additional identifying information. By setting
$G$ to be bounded from above, we dampen the influence of large values
of $\g_{\tJ,t,s}({\bf x}_{ts})$, so that the numerical minimization
of $\hat{Q}$ is less sensitive to potentially large but redundant
variations in $\hat{\g}_{\tJ,t,s}({\bf x}_{ts})$.

\subsubsection*{Angle-Space Reparameterization of $\protect\S^{D-1}$}

To minimize $\hat{Q}(\b)$ over $\b\in\S^{D-1}$, we work with a reparameterization
of $\S^{D-1}$ with $D-1$ angles in spherical coordinates.\footnote{The idea and the motivation for using the angle-space reparameterization
can also be found in \citet*{manski1986operational}, who however
use only one angle parameter.} Specifically, define the angle space $\T$ by
\begin{equation}
\T:=\left[-\pi,\pi\right)\times\left[-\frac{\pi}{2},\frac{\pi}{2}\right]^{D-2},\label{eq:Theta}
\end{equation}
and the transformation $\t\longmapsto\b(\t)$ by standard spherical
coordinate transformation. We now instead solve the optimization of
$\hat{Q}(\b(\t))$ over $\T$, which we further equip with its natural
geodesic metric $\rho_{\T}\left(\t,\tilde{\t}\right):=\arccos\left(\b\left(\t\right){}^{'}\b\left(\tilde{\t}\right)\right)$.
Note that $\rho_{\T}\left(\t,\tilde{\t}\right)$ is strongly equivalent\footnote{Two metrics $d_{1}$ and $d_{2}$ defined on some nonempty set $X$
are\emph{ strongly equivalent} if and only if there exist positive
constants $c_{1}$ and $c_{2}$ such that $c_{1}d_{1}\left(x,y\right)\leq d_{2}\left(x,y\right)\leq c_{2}d_{1}\left(x,y\right)$
for every $x,y\in X$.} to the (imported) Euclidean distance $\norm{\b\left(\t\right)-\b\left(\tilde{\t}\right)}$.

This reparameterization $\left(\T,\rho_{\T}\right)$ enables us to
exploit the compactness and convexity of the parameter space $\T=\left[-\pi,\pi\right)\times\left[-\frac{\pi}{2},\frac{\pi}{2}\right]^{D-2}$,
which takes the form of a hyper-rectangle. First, $\left(\T,\rho_{\T}\right)$
preserves all topological structures of the unit sphere, and particularly
inherits the compactness of $\left(\S^{D-1},\norm{\cd}\right)$, automatically
satisfying the compactness condition usually imposed for extremum
estimation and making it numerically feasible to initiate a grid on
the whole parameter space. Second, while the unit sphere $\S^{D-1}$
is not convex, the new parameter space $\T$ becomes convex algebraically,
making it computationally easy to define bisection points in the parameter
space. Third, $\left(\T,\rho_{\T}\right)$ preserves the geometric
structures of the sphere, including, for instance, the obvious observation
that $-\pi$ and $\pi$ in the first coordinate of $\T$ should be
treated as exactly the same point, or more rigorously, $\rho_{\T}\left(\left(\pi-\e,\t_{2},\ldots,\t_{D-1}\right),\left(-\pi,\t_{2},\ldots,\t_{D-1}\right)\right)\to0$
as $\e\to0$. This seemingly trivial property is nevertheless important
in defining and interpreting whether certain parameter estimates converge
asymptotically or not.

\subsubsection*{An Adaptive-Grid Algorithm}

With the angle reparameterization, we seek to numerically compute
a conservative rectangular enclosure of $\arg\min\hat{Q}\left(\t\right)$,
deploying a bisection-style\emph{ }grid-search algorithm that recursively
shrinks and refines an \emph{adaptive grid} to any pre-chosen precision
(as defined by $\rho_{\T}$). Unlike gradient-based local optimization
algorithms, our adaptive grid algorithm handles the built-in discreteness
in our sample criterion function, whose derivative is zero almost
everywhere, while still maintaining global coverage over the entire
parameter space. While a brute-force global search algorithm is the
safest choice when the dimension of the product characteristics $D$
is relatively small, our adaptive-grid algorithm runs significantly
faster. The essential structure of our algorithm is laid out as follows.

\medskip{}

Step 1: Initialize a global grid $\T^{\left(1\right)}$ of some chosen
size $M_{0}^{D-1}$ on $\T$.

Step 2: Compute $\hat{Q}\left(\t\right)$ for each $\t\in\T^{\left(1\right)}$,
and select all points in $\T^{\left(1\right)}$ with a criterion value
below the $\a$th-quantile in $\hat{Q}\left(\T^{\left(1\right)}\right):=\left\{ \hat{Q}\left(\t\right):\t\in\T^{\left(1\right)}\right\} $
into
\begin{equation}
\ul{\T}^{\left(1\right)}:=\left\{ \t\in\T^{\left(1\right)}:\ \hat{Q}\left(\t\right)\leq\mathrm{quantile}_{\a}\left(\hat{Q}\left(\T^{\left(1\right)}\right)\right)\right\} .\label{eq:Theta_ul1}
\end{equation}

Step 3: Take the enclosing rectangle of $\ul{\T}^{\left(1\right)}$,
by defining $\ul{\t}_{d}^{\left(1\right)}:=\mathrm{min}^{*}\ul{\T}_{d}^{\left(1\right)}$
and $\ol{\t}_{d}^{\left(1\right)}:=\mathrm{max}^{*}\ul{\T}_{d}^{\left(1\right)},$
where $\ul{\T}_{d}^{\left(1\right)}:=\left\{ \t_{d}:\t\in\ul{\T}^{\left(1\right)}\right\} $
for each $d=1,\ldots,D-1$ and the operator $\mathrm{min}^{*}$ and
$\mathrm{max}^{*}$ have standard definitions of $\min$ and $\max$
except for the first dimension $d=1$. For the first dimension, it
is necessary to account for the underlying spherical geometry and
the periodicity of angles, i.e. $\t_{1}+2\pi\equiv\t_{1}$ and in
particular $-\pi\equiv\pi$. This, however, is largely a programming
nuisance: whenever $\ul{\T}_{1}^{\left(1\right)}\subsetneq\T_{1}^{\left(1\right)}$
crosses over at $-\pi$ and $\pi$, we can add $2\pi$ to every $\t_{1}\in\ul{\T}_{1}^{\left(1\right)}$
and obtain lower and upper bounds of $\ul{\T}_{1}^{\left(1\right)}+2\pi$,
as illustrated in Figure \ref{fig:Algo}.

Step 4: We initialize a refined grid $\T^{\left(2\right)}$ on $\ol{\ul{\T}}^{\left(1\right)}:=\times_{d=1}^{D-1}\left[\ul{\t}_{d}^{\left(1\right)},\ol{\t}_{d}^{\left(1\right)}\right]$
of size $M_{0}^{D-1}$.

Step 5: Iterate until refinement stops (falls below a certain numerical
precision).

\medskip{}

Note that the above is simply a sketch of our algorithm: see Appendix
\ref{sec:app_GridSearch} and the documentation on GitHub for more
implementation details.\footnote{Our algorithm relies heavily on the compactness and convexity of the
angle space $\T$. Compactness allows us to start with a global grid
over the whole parameter space for initial evaluations of the sample
criterion function. At each step of recursion, the convexity of $\T$
enables us to conveniently refine the grid by separately cutting each
coordinate of $\ol{\ul{\T}}^{\left(m\right)}$ into smaller pieces
through simple division.} To be conservative, we add in buffers at each step of refinement,
keep track of both outer and inner boundaries of the lower-quantile
set $\ul{\T}^{\left(m\right)}$, and make sure that the minimizers
of the criterion functions at all computed points are indeed enclosed
by the set returned in the end. We find the current algorithm to be
conservative and to perform well in our simulations.

This multi-stage approach is designed to balance computational feasibility
with a robust search. The initial coarse search efficiently discards
large, suboptimal regions of the parameter space, while the subsequent
refinement and boundary identification stages provide a high-precision
estimate in the most promising area. However, the algorithm's performance
is inherently tied to the selection of tuning parameters, particularly
the initial grid size \texttt{M\_Step} and the quantile used for pruning,
which must be chosen carefully to ensure the global minimum is not
discarded prematurely. Furthermore, it can get computationally intensive
when the dimension of $\beta$ is high. We find that running 1,000
simulations for $D=3$ usually takes a few hours on modern computers,
while for $D=4$ it may take up to a day.

\section{\label{sec:Ext_MMIM}General Econometric Framework of Multi-Index
Single-Crossing Conditions}

Our key identification strategy, and consequently the associated estimation
method, apply more widely beyond panel multinomial choice models.
We now introduce a general econometric framework defined by \emph{multi-index
single-crossing} (MISC) conditions, and show how our proposed methods
can be exploited in a wide range of models nested under the MISC condition
framework.

Formally, let $\left(y_{i},X_{i}\right)_{i=1}^{n}$ be a random sample
of data with $X_{i}$ distributed on the support ${\cal X}\subseteq\R^{d_{x}}$
and $y_{i}$ distributed on ${\cal {\cal Y}}\subseteq\R^{d_{y}}$.
Let $h_{0}:{\cal X}\to\R$ denote a functional of the conditional
distribution of $y_{i}$ given $X_{i}$ that is directly identified
from data. For each of $j=1,\ldots,J\in\N$, let $\phi_{j}:{\cal X}\to\R^{d_{\t_{j}}}$
be some known transformation of $X_{i}$, and define $W_{ij}:=\phi_{j}\left(X_{i}\right)$
with $W_{i}:=\left(W_{i1},\ldots,W_{iJ}\right)$. Let $\t_{0j}\in\T_{j}\subseteq\R^{d_{\t_{j}}}$
be an unknown finite-dimensional parameter and write $\t_{0}:=\left(\t_{01}^{'},\ldots,\t_{0J}^{'}\right)^{'}\in\T:=\times_{j=1}^{J}\T_{j}$.
\begin{defn}[\emph{Multi-Index Single-Crossing Condition}]
\label{assu:Assum_Mono} We say that $\left(h_{0},\t_{0}\right)$
satisfy the (\emph{weak})\emph{ multi-index single-crossing condition
}if, for any realization $x\in{\cal X}$ and $w=\phi\left(x\right)$,
\begin{align}
w_{j}^{'}\t_{0j}\geq0,\ \forall j=1,\ldots,J\quad & \imp\quad h_{0}\left(x\right)\geq0,\nonumber \\
w_{j}^{'}\t_{0j}\leq0,\ \forall j=1,\ldots,J\quad & \imp\quad h_{0}\left(x\right)\leq0.\label{eq:MISC}
\end{align}
The condition is said to be strict if the inequalities on the right-hand
side of \eqref{eq:MISC} are strict.
\end{defn}
\noindent In words, the MISC condition states that if all the $J$
parametric indices $w_{1}^{'}\t_{01}$, $w_{2}^{'}\t_{02}$, $\ldots$,
and $w_{J}^{'}\t_{0J}$ are (weakly) positive, then the functional
$h_{0}$ must be (weakly) positive; if the $J$ indices are all zero,
then $h_{0}$ must be zero; if the $J$ indices are all negative,
then $h_{0}$ must be negative. Essentially, the MISC condition provides
a parsimonious way to semiparametrically model how the multiple economic
factors jointly affect a certain statistic of the relevant economic
outcome. The MISC condition basically requires that, if all the relevant
factors reach certain thresholds, then the outcome statistics must
also reach certain thresholds. Such requirements are often easy to
obtain in an economic or econometric model: while multiple factors
in an economic model may interact with each other in potentially complicated
manners and there might be many configurations of the factors that
lead to ambiguous theoretical predictions, there are also often simple
configurations that we understand reasonably well. Hence, the MISC
condition imposes only mild requirements on the underlying economic
or econometric model for the problem and thus provides a general framework
for semiparametric econometric analysis, in which most of the modeling
ingredients can be left nonparametric except for the parametric indices
that capture different economic factors in the problem.

Clearly, the panel multinomial choice model considered in previous
sections falls under the MISC condition framework. Specifically, focusing
on a pair of time periods $\left(t,s\right)$ and a particular product
$j_{0}$ for illustration, define $\t_{0j}:=\b_{0}$, $h_{0}\left({\bf X}_{i}\right):=\g_{j_{0},ts}\left({\bf X}_{i}\right)$,
$W_{ij_{0}}:=X_{ij_{0}t}-X_{ij_{0}s}$ and $W_{ij}:=-\left(X_{ijt}-X_{ijs}\right)$
for $j\neq j_{0}$. Then, the MISC condition \eqref{eq:MISC} is satisfied
under model \eqref{eq:Model_PMC} and Assumptions \ref{assu:RandSamp}--\ref{assu:EpsDist}.

We now provide a few more examples of models nested in the MISC condition
framework.
\begin{example}[Binary Choice with Awareness]
\label{exa:Bin_Aware} Consider the following binary choice model
\[
y_{i}=\ind\left\{ X_{i1}^{'}\t_{01}\geq u_{i}\right\} \cd\ind\left\{ X_{i2}^{'}\t_{02}\geq v_{i}\right\} 
\]
where $y_{i}$ denotes whether consumer $i$ purchases a certain product
or not, $X_{i1}$ denotes a vector of covariates that influences the
consumer's utility from a product, and $X_{i2}$ denotes a vector
of covariates that affects the consumer's awareness of the product
(e.g., advertising). Here, we have $J=2$, $X_{i}:=\left(X_{i1},X_{i2}\right)$,
$W_{i1}:=X_{i1}$, and $W_{i2}:=X_{i2}$. Define $h_{0}\left(x\right):=\E\left[\rest{y_{i}}X_{i}=x\right]-\frac{1}{4}.$
Then, under the conditional median restrictions $\text{med}\left(\rest{u_{i}}X_{i}\right)=\text{med}\left(\rest{v_{i}}X_{i}\right)=0$
and the conditional independence restriction $\rest{u_{i}\indep v_{i}}X_{i}$,
it is true that
\begin{align*}
X_{i1}^{'}\t_{01}>0,\ X_{i2}^{'}\t_{02}>0 & \quad\imp\quad h_{0}\left(X_{i}\right)>0,\\
X_{i1}^{'}\t_{01}<0,\ X_{i2}^{'}\t_{02}<0 & \quad\imp\quad h_{0}\left(X_{i}\right)<0,
\end{align*}
satisfying the MISC condition.
\end{example}
\begin{example}[Binary Choice with Endogeneity]
\label{exa:Bin_Choice_Endo} Consider the binary choice model
\[
Y_{i}=\ind\left\{ W_{i}^{'}\b_{0}\geq\e_{i}\right\} ,
\]
and let one component of $W_{i}$, say, $W_{i1}$ be endogenous. Suppose
that there exists a vector of instrumental variables $Z_{i}$ and
define $\xi_{i}:=W_{i1}-Z_{i}^{'}\g_{0}$ as the residual from the
reduced-form linear projection of $W_{i1}$ on $Z_{i}$. Assume that
the endogeneity between $\e_{i}$ and $W_{i1}$ is captured by the
following control function
\[
\text{med}\left(\rest{\e_{i}}Z_{i},\xi_{i}\right)=\l\left(\a_{0}\xi_{i}\right),
\]
where $\l$ is an unknown increasing function with location normalization
$\l\left(0\right)=0$, and the sign parameter $\a_{0}\in\left\{ -1,1\right\} $
controls the direction of the monotonicity. The above can be viewed
as an adaptation of the binary choice model that combines the conditional
median restriction in \citet*{manski1975maximum} with the control
function approach in \citet*{blundell2004endogeneity}: here we only
impose the control function restriction on the conditional median
instead of the whole distribution as in \citet*{blundell2004endogeneity}.
Then, writing $\ol Z_{i}:=\left(W_{i1},Z_{i}\right)$ and $\ol{\g}_{0}:=\left(-\a_{0},\a_{0}\g_{0}\right)^{'}$,
we have 
\begin{align*}
W_{i}^{'}\b_{0}>0,\ \ol Z_{i}^{'}\ol{\g}_{0}>0\quad\imp\quad\E\left[\rest{Y_{i}}W_{i},Z_{i}\right] & >\frac{1}{2}
\end{align*}
and its ``$<$'' counterpart, which can be viewed as a MISC condition
with $K=2,$ $h_{0}\left(W_{i},Z_{i}\right):=\E\left[\rest{Y_{i}-\frac{1}{2}}W_{i},Z_{i}\right]$,
$\phi_{1}\left(W_{i},Z_{i}\right):=W_{i}$, $\phi_{2}\left(W_{i},Z_{i}\right):=\left(W_{i1},Z_{i}\right)$,
$\t_{0,1}:=\b_{0}$, and $\t_{0,2}:=\ol{\g}_{0}$.
\end{example}
\begin{example}[Dyadic Network Formation]
\label{exa:NetForm} Consider the dyadic network formation model
of \citet*{gao2023logical}, which extends \citet{graham2017econometric}
to a semiparametric setting:
\begin{align*}
\E\left[\rest{y_{ij}}X_{i},X_{j},A_{i},A_{j}\right]\ =\  & \psi\left(w\left(X_{i},X_{j}\right)^{'}\t_{0},A_{i},A_{j}\right).
\end{align*}
Here $y_{ij}$ is a binary outcome indicating whether individuals
$i$ and $j$ are linked in an undirected network, $X_{i}$ and $X_{j}$
are the individuals' observable covariates, $w(X_{i},X_{j})$ is a
known pairwise transformation of individual covariates (with the leading
example being $w_{h}\left(X_{i},X_{j}\right):=\left|X_{i,h}-X_{j,h}\right|$
for each coordinate $h=1,\ldots,d_{x}$), $A_{i}$ and $A_{j}$ are
unobserved individual degree heterogeneity terms, and $\psi:\R^{3}\to\R$
is an unknown function assumed to be increasing in all its three arguments.
Specifically, fixing a particular pair of individuals $(\ol i,\ol j)$
and two realizations $\ol x,\ul x$ of $X_{i}$, it can be shown that,
with
\[
\ol w:=w\left(x_{\ol j},\ol x\right)-w\left(x_{\ol i},\ol x\right),\quad\ul w:=w\left(x_{\ol i},\ul x\right)-w\left(x_{\ol j},\ul x\right),
\]
and
\begin{align*}
h_{0}\left(\ol x,\ul x\right):= & \max\left(0,\E\left[\rest{y_{\ol ik}-y_{\ol jk}}X_{k}=\ol x\right]\right)\E\left[\rest{y_{\ol ik}-y_{\ol jk}}X_{k}=\ul x\right]\\
 & -\max\left(0,\E\left[\rest{y_{\ol jk}-y_{\ol ik}}X_{k}=\ol x\right]\right)\E\left[\rest{y_{\ol jk}-y_{\ol ik}}X_{k}=\ul x\right]
\end{align*}
the weak MISC condition is satisfied under mild conditions:
\begin{align*}
\ol w^{'}\t_{0}>0,\ \ul w^{'}\t_{0}>0\quad & \imp\quad h_{0}\left(\ol x,\ul x\right)\geq0,\\
\ol w^{'}\t_{0}<0,\ \ul w^{'}\t_{0}<0\quad & \imp\quad h_{0}\left(\ol x,\ul x\right)\leq0.
\end{align*}
\end{example}
\begin{example}[Censored Monotone Transformation Model with Endogeneity]
\label{exa:CensorEndo} The approach proposed in Example \ref{exa:Bin_Choice_Endo}
above can also be adapted to the following censored monotone transformation
model with endogeneity:
\[
Y_{i}=\max\left\{ \phi\left(W_{i}^{'}\b_{0},\e_{i}\right),0\right\} ,
\]
where one component of the observed covariates, $W_{i1}$, is endogenous,
and $\phi$ is an unknown bivariate increasing function. This model
generalizes the usual censored regression model $Y_{i}=\max\left\{ W_{i}^{'}\b_{0}+\e_{i},0\right\} $,
say, in \citet*{blundell2007censored}, by incorporating a flexible
unknown monotone transformation $\phi$ with non-additive error term.
Since $\b_{0}$, $\e_{i}$ and $\phi$ are all unknown, one may normalize
$\phi\left(0,0\right)=0$. By the equivariance of (conditional) quantiles
under monotone transformations, we have 
\[
\text{med}\left(\rest{Y_{i}}W_{i},Z_{i}\right)=\max\left\{ \phi\left(W_{i}^{'}\b_{0},\text{med}\left(\rest{\e_{i}}W_{i},Z_{i}\right)\right),0\right\} .
\]
Similar to Example \ref{exa:Bin_Choice_Endo}, define $\xi_{i}:=W_{i1}-Z_{i}^{'}\g_{0}$
as the residual from the reduced-form linear projection of $W_{i1}$
on the instrumental variables $Z_{i}$, and assume that the endogeneity
between $\e_{i}$ and $W_{i1}$ is captured by the control function
$\text{med}\left(\rest{\e_{i}}W_{i},Z_{i}\right)=\text{med}\left(\rest{\e_{i}}Z_{i},\xi_{i}\right)=\l\left(\a_{0}\xi_{i}\right)$,
where $\l$ is an increasing function with normalization $\l\left(0\right)=0$.
Writing $\ol Z_{i}:=\left(W_{i1},Z_{i}\right)$ and $\ol{\g}_{0}:=\left(\a_{0},-\a_{0}\g_{0}\right)^{'}$,
we have 
\begin{align*}
W_{i}^{'}\b_{0}>0,\ \ol Z_{i}^{'}\ol{\g}_{0}>0\quad\imp\quad\text{med}\left(\rest{Y_{i}}W_{i},Z_{i}\right) & >0\text{ and}\\
W_{i}^{'}\b_{0}\leq0,\ \ol Z_{i}^{'}\ol{\g}_{0}\leq0\quad\imp\quad\text{med}\left(\rest{Y_{i}}W_{i},Z_{i}\right) & =0,
\end{align*}
which can be viewed as a MISC condition with a weak ``$\leq$''
side and $h_{0}\left(X_{i}\right):=\text{med}\left(\rest{Y_{i}}W_{i},Z_{i}\right)$
given by the conditional median function. 
\end{example}
Based on the MISC conditions \eqref{eq:MISC}, we can again obtain
identifying restrictions by taking their logical contrapositions,
which can be encoded algebraically in a similar way as in \eqref{prop:ID_rest}
and \eqref{eq:Q_G_jts}. Specifically, let $G$ be a one-sided sign-preserving
function as in \eqref{eq:Q_G_jts} and define
\[
\l\left(W_{i};\t\right):=\prod_{j=1}^{J}\ind\left\{ W_{ij}^{'}\t_{j}\leq0\right\} .
\]

\begin{prop}
\label{prop:ID_Gen} Under condition \eqref{eq:MISC}, we have
\begin{align*}
h_{0}\left(X_{i}\right)>0 & \ \imp\ \text{NOT}\ \left\{ W_{ij}^{'}\t_{j}\leq0\ \forall j\right\} ,\\
h_{0}\left(X_{i}\right)<0 & \ \imp\ \text{NOT}\ \left\{ W_{ij}^{'}\t_{j}\geq0\ \forall j\right\} .
\end{align*}
Furthermore, with $Q\left(\t\right):=Q_{+}\left(\t\right)+Q_{-}\left(\t\right)$
where
\begin{align*}
Q_{+}\left(\t\right) & :=\E\left[G\left(h_{0}\left(X_{i}\right)\right)\l\left(W_{i};\t\right)\right]\ \text{ and }\ Q_{-}\left(\t\right):=\E\left[G\left(-h_{0}\left(X_{i}\right)\right)\l\left(-W_{i};\t\right)\right],
\end{align*}
we have $Q\left(\t\right)\geq Q\left(\t_{0}\right)=0$.
\end{prop}
\noindent Proposition \ref{prop:ID_Gen} generalizes Theorem \ref{thm:SetID}.
Notice that Proposition \ref{prop:ID_Gen} applies to all functionals
$h_{0}$ of the conditional distribution of $y_{i}$ given $\X_{i}$
that satisfy the MISC conditions.

One could also proceed with the two-step estimation procedure described
in Section \ref{sec:S_EstComp}. Given a first-stage nonparametric
estimator $\hat{h}$ of $h_{0}$, we can estimate $\t_{0}$ (or the
identified set) by minimizing the sample criterion $\hat{Q}\left(\t\right):=\hat{Q}_{+}\left(\t\right)+\hat{Q}_{-}\left(\t\right)$
with
\[
\hat{Q}_{+}\left(\t\right):=\frac{1}{n}\sum_{i=1}^{n}G\left(\hat{h}\left(X_{i}\right)\right)\l\left(W_{i};\t\right)\ \text{ and }\ \widehat{Q}_{-}\left(\t\right):=\frac{1}{n}\sum_{i=1}^{n}G\left(-\hat{h}\left(X_{i}\right)\right)\l\left(-W_{i};\t\right).
\]

\section{\label{sec:Sim}Simulation}

We now switch back to the panel multinomial choice model introduced
in Section \ref{sec:PMC} and examine the finite sample performance
of our proposed estimator. For each DGP, we run $M=1,000$ simulations
of model \eqref{eq:Model_PMC} with the following utility specification:
\[
u\left(X_{ijt}^{'}\b_{0},\,A_{ij},\,\e_{ijt}\right)=A_{i0}\left(X_{ijt}^{'}\b_{0}+A_{ij}\right)+\e_{ijt},
\]
in which $A_{i0}$ is an unobserved scale fixed effect that captures
agent-level heteroskedasticity in utilities, and $A_{ij}$ is an unobserved
location shifter specific to each agent-product pair. The ability
to deal with nonlinear dependence caused by unobserved fixed effects
in a relatively robust way is a distinctive feature of our method
compared with existing approaches. To allow for such dependence, we
generate correlation between the observable characteristics ${\bf X}_{i}$
and the fixed effects ${\bf A}_{i}$ via a latent variable $Z$. We
draw $Z_{i}\sim\mathcal{N}$$\left(0,1\right)$ and let $A_{i2}=\left[Z_{i}\right]_{+}$.
We construct $X_{ijt,2}=W_{ijt}+Z_{i}$ with $W_{ijt}\sim\cN\left(0,2J\right)$.
Thus, $X$ and $A$ are correlated via $Z$. The DGPs for the rest
of ${\bf A}$ and ${\bf X}$ are: $A_{i0}\sim\mathcal{U}\left[2,2.5\right]$,
$A_{i1}\equiv0$, $A_{ij}\sim\mathcal{U}\left[-0.25,0.25\right]$
for $j\geq3$, $X_{ijt,1}\sim\mathcal{U}\left[-1,1\right]$, $X_{ijt,d}\sim\cN\left(0,1\right)$
for $d\geq3$. Furthermore, we set $\ol{\b}_{0}=\left(2,1,\ldots,1\right)^{'}\in\R^{D}$
and $\b_{0}=\ol{\b}_{0}/\norm{\ol{\b}_{0}}$, and draw $\epsilon_{ijt}\sim TIEV\left(0,1\right)$.
To summarize, for each of the $M=1,000$ simulations we first generate
$\left(\b_{0},\X_{it},{\bf A}_{i},\boldsymbol{\epsilon}_{it}\right)$
for all $(i,t)$ pairs. Then, we calculate the individual choice ${\bf Y}$
matrix according to model \eqref{eq:Model_PMC}. Next, we compute
$\hat{\b}$ from the simulated observable data of $\left({\bf X},{\bf Y}\right)$.
To obtain $\hat{\b}$, we first use nonparametric regression with
second-order polynomial basis functions with $\ell_{1}$-regularization
and 10-fold cross validation to estimate $\gamma$. Then, we apply
the adaptive-grid algorithm detailed in Section \ref{subsubsec:Comp}.
Finally, we assess how well $\hat{\b}$ performs compared with the
true value $\b_{0}$.\footnote{In Appendix \ref{sec:App_ADDSIM}, we provide additional simulation
results. Specifically, we first present a graphical illustration of
the identified set $B_{0}$ based on the population criterion \eqref{eq:PMC_BigQ}.
Second, we inspect how our estimator performs without point identification.
Third, we vary $\left(D,J,T\right)$ to examine how robust our method
is against various simulation specifications. Lastly, we include a
simulation illustration of the robustness of our approach to the ``Blue-Bus/Red-Bus''
problem.}

\subsubsection*{Baseline Results}

For the baseline configuration, we set $N=10,000,\ D=3,\ J=3,\text{ and }T=2$.
Since in this case the conditions for the point identification are
satisfied, any point from the argmin set $\hat{B}_{b}:=\arg\min_{\b\in\S^{D-1}}\hat{Q}_{b}\left(\b\right)$
is a consistent estimator of $\b_{0}$ for each round of simulation
$b=1,\ldots,M$. Specifically, we define
\[
\hat{\b}_{b,d}^{u}:=\max\hat{B}_{b,d},\quad\hat{\b}_{b,d}^{l}:=\min\hat{B}_{b,d},\quad\text{and}\quad\hat{\b}_{b,d}^{m}:=\frac{1}{2}\left(\hat{\b}_{b,d}^{u}+\hat{\b}_{b,d}^{l}\right),
\]
where $\hat{\b}_{b,d}^{u}$, $\hat{\b}_{b,d}^{l}$, and $\hat{\b}_{b,d}^{m}$
represent the maximum, minimum, and middle point along dimension $d$
for each round of simulation $b$ of the argmin set $\hat{B}$, respectively. 

\begin{table}
\caption{Baseline Performance\label{tab:Baseline-Estimation-Performance}}

\bigskip{}

\noindent \centering{}%
\begin{tabular}{ccccc}
\toprule 
$\phantom{\frac{\frac{1}{1}}{\frac{1}{1}}}$ & $\b_{0}=\left(0.82,0.41,0.41\right)^{'}$ & $\hat{\b}_{1}$ & $\hat{\b}_{2}$ & $\hat{\b}_{3}$\tabularnewline
\midrule 
$\phantom{\frac{\frac{1}{1}}{\frac{1}{1}}}$mid bias & $\frac{1}{M}\sum_{b=1}^{M}\left(\hat{\b}_{b,d}^{m}-\b_{0,d}\right)$ & -0.0005 & -0.0003 & -0.0034\tabularnewline
$\phantom{\frac{\frac{1}{1}}{\frac{1}{1}}}$upper bias & $\frac{1}{M}\sum_{b=1}^{M}\left(\hat{\b}_{b,d}^{u}-\b_{0,d}\right)$ & 0.0075 & 0.0080 & 0.0083\tabularnewline
$\phantom{\frac{\frac{1}{1}}{\frac{1}{1}}}$lower bias & $\frac{1}{M}\sum_{b=1}^{M}\left(\hat{\b}_{b,d}^{l}-\b_{0,d}\right)$ & -0.0085 & -0.0086 & -0.0150\tabularnewline
$\phantom{\frac{\frac{1}{1}}{\frac{1}{1}}}$mean(u$-$l) & $\frac{1}{M}\sum_{b=1}^{M}\left(\hat{\b}_{b,d}^{u}-\hat{\b}_{b,d}^{l}\right)$ & 0.0160 & 0.0166 & 0.0233\tabularnewline
standard deviation & $\sqrt{\frac{1}{M}\sum_{b=1}^{M}\left(\hat{\b}_{b,d}^{m}-\ol{\hat{\b}_{d}^{m}}\right)^{2}}$ & 0.0299 & 0.0308 & 0.0431\tabularnewline
root MSE (by coordinate) & \textcolor{black}{$\left(\frac{1}{M}\sum_{b=1}^{M}\left(\hat{\b}_{b,d}^{m}-\b_{0,d}\right)^{2}\right)^{1/2}$} & 0.0288 & 0.0296 & 0.0417\tabularnewline
\midrule 
\textcolor{black}{$\phantom{\frac{\frac{1}{1}}{\frac{1}{1}}}$root
MSE (whole vector)} & \textcolor{black}{$\left(\frac{1}{M}\sum_{b=1}^{M}\norm{\hat{\b}_{b}^{m}-\b_{0}}^{2}\right)^{1/2}$} & \multicolumn{3}{c}{\textcolor{black}{0.0587}}\tabularnewline
\textcolor{black}{$\phantom{\frac{\frac{1}{1}}{\frac{1}{1}}}$}$\begin{array}{c}
\text{mean norm}\\
\text{deviations (MND) }
\end{array}$ & \multirow{1}{*}{$\frac{1}{M}\sum_{b=1}^{M}\norm{\hat{\b}_{b}^{m}-\b_{0}}$} & \multicolumn{3}{c}{0.0511}\tabularnewline
\bottomrule
\end{tabular}
\end{table}

Table \ref{tab:Baseline-Estimation-Performance} summarizes our baseline
results. In the first row we use the middle point $\hat{\b}^{m}$
along each dimension of $\hat{B}$ to calculate the bias. The biases
are very small across all three dimensions with a magnitude between
-0.0034 and -0.0005. The next two rows show the biases in estimating
$\b_{0,d}$ using $\hat{\b}_{d}^{u}$ and $\hat{\b}_{d}^{l}$ respectively,
which are again close to zero. The fourth row reports the average
widths of the set $\hat{B}$ along each dimension. These widths are
small relative to the magnitude of $\b_{0}$. The fifth and sixth
rows summarize the standard deviation and rMSE for each coordinate
of $\widehat{\b}^{m}$. In the second part of Table \ref{tab:Baseline-Estimation-Performance},
we report the vector rMSE and MND based on $\hat{\b}^{m}$, and the
results suggest that our method performs well.

\subsubsection*{Results Varying $N$}

Next, we vary $N$ while maintaining $D=3,\ J=3,\ \text{and }T=2$
to assess how our method performs under different sample sizes. In
addition to the baseline setup with $N=10,000$, we calculate mean
absolute deviation (MAD), average size of the estimated set, rMSE,
and MND for $N=4,000$ and $N=1,000$. Results are summarized in Table
\ref{tab:PerformanceVaryingN}.

\begin{table}
\caption{Performance under Varying $N$\label{tab:PerformanceVaryingN}}

\bigskip{}

\noindent \centering{}%
\begin{tabular}{ccccc}
\toprule 
\multirow{2}{*}{} & \multirow{2}{*}{$\sum_{d}\left|\text{bias}_{d}\right|$} & \multirow{2}{*}{$\sum_{d}\text{mean(u-l)}_{d}$} & \multirow{2}{*}{\textcolor{black}{$\text{rMSE}$}} & \multirow{2}{*}{MND}\tabularnewline
 &  &  &  & \tabularnewline
\midrule 
$\phantom{\frac{\frac{1}{1}}{\frac{1}{1}}}$$N=10,000$ & 0.0042 & 0.0560 & 0.0587 & 0.0511\tabularnewline
$\phantom{\frac{\frac{1}{1}}{\frac{1}{1}}}$$N=\ 4,000$ & 0.0136 & 0.0865 & \textcolor{black}{0.0742} & 0.0650\tabularnewline
$\phantom{\frac{\frac{1}{1}}{\frac{1}{1}}}$$N=\ 1,000$ & 0.0606 & 0.1664 & \textcolor{black}{0.1369} & 0.1159\tabularnewline
\midrule
$\phantom{\frac{\frac{1}{1}}{\frac{1}{1}}}$ & $\left(\dfrac{N}{1,000}\right)^{1/2}$ & $\left(\dfrac{N}{1,000}\right)^{1/3}$ & $\dfrac{\text{rMSE}_{1000}}{\text{rMSE}_{N}}$ & $\dfrac{\text{MND}_{1000}}{\text{MND}_{N}}$\tabularnewline
\midrule
$\phantom{\frac{\frac{1}{1}}{\frac{1}{1}}}$\textcolor{blue}{${\color{black}N=10,000}$} & 3.16 & 2.15 & 2.33 & 2.27\tabularnewline
$\phantom{\frac{\frac{1}{1}}{\frac{1}{1}}}$$N=\ 4,000$ & 2.00 & 1.59 & 1.84 & 1.78\tabularnewline
\bottomrule
\end{tabular}
\end{table}

Table \ref{tab:PerformanceVaryingN} provides numerical evidence that
a larger $N$ helps with overall performance. The sum of absolute
bias decreases from 0.0606 to 0.0042 when $N$ increases from $1,000$
to $10,000$. The average size of the estimated sets, rMSE, and MND
follow a similar pattern. Notably, even with a relatively small $N=1,000$,
the results remain informative and reasonably accurate, with the rMSE
and MND equal to 0.1369 and 0.1159, respectively. We note that $T$
is set to 2 here, which is the minimum required for our method to
work. Since our method can extract information from each of the $T\left(T-1\right)$
ordered pairs of time periods, a larger $T$ would generally improve
the performance of our estimators. Appendix \ref{sec:App_ADDSIM}
presents additional simulation results for a larger $T$.

Finally, we numerically investigate the speed of convergence when
we increase $N$ from $1,000$ to $4,000$ and $10,000$ in the second
part of Table \ref{tab:PerformanceVaryingN}. Compared with the case
of $N_{0}=1,000$, the relative ratios of rMSE are 1.84 for $N=4,000$
and 2.33 for $N=10,000$, both of which lie between $\left(N/N_{0}\right)^{1/3}$
and $\left(N/N_{0}\right)^{1/2}$. A similar pattern is also observed
for calculations based on MND. These results suggest that our estimator
converges at a rate slower than $N^{-1/2}$ but faster than $N^{-1/3}$.

\section{\label{sec:Emp}Empirical Application}

\subsection{\label{subsec:Data}Data and Methodology}

We now present an empirical application of the panel multinomial choice
model and our proposed estimation method, using NielsenIQ Retail Scanner
Data on popcorn sales to examine the effects of display promotions.
The data contain weekly store-level information on prices, sales,
and display promotion status, collected from approximately 35,000
participating retail stores with point-of-sale systems across the
United States.

We focus on popcorn among the wide array of products for two reasons.
First, popcorn purchases are more likely to be impulsive, with limited
intertemporal planning. Second, popcorn exhibits substantial variation
in in-store display promotions, allowing us to estimate how special
displays influence consumers\textquoteright{} purchase decisions.

We aggregate store-level observations to the designated market area
(DMA) level ($N=205$) for 2015. We focus on the top three brands
by market share, pool the remaining brands into a fourth product (``all
other products''), and include an outside option of ``no purchase.''
We compute market shares as the dependent variable for each of the
$J=5$ alternatives---the three leading brands, the ``all other products''
category, and the outside option. The observed product characteristics
include price, display-promotion status, and their interaction.\footnote{We define $\text{Price}_{cjt}$ as the weighted-average unit price
of all UPCs of the brand $j$ in DMA $c$ during week $t$. The dataset
includes two promotion indicators: display and feature. Given their
similarity, we construct $\text{Promo}_{cjt}$ as $($feature$\lor$display$)_{cjt}$.
The interaction term $\text{Price}_{cjt}\times\text{Promo}_{cjt}$
is included in $X$ to allow price sensitivity (elasticity) to vary
under promotion.} Notationally, $c$ denotes each of the $N=205$ DMAs, $j$ represents
each of the $J=5$ brands, and $t$ indexes the $T=52$ weeks in 2015.
The summary statistics of these variables are provided in Table \ref{tab:Empirical-Application:-Summary}.

\begin{table}
\caption{Empirical Application: Summary Statistics\label{tab:Empirical-Application:-Summary}}

\bigskip{}

\noindent \centering{}%
\begin{tabular}{ccccc}
\toprule 
\textcolor{black}{$\phantom{\frac{\frac{1}{1}}{\frac{1}{1}}}$} & mean & s.d. & min & max\tabularnewline
\midrule
\textcolor{black}{$\phantom{\frac{\frac{1}{1}}{\frac{1}{1}}}$}$\text{DMA-level\ Market\ Share}$
$s_{cjt}$ & $25.06\%$ & $21.59\%$ & $0.08\%$ & $96.69\%$\tabularnewline
\textcolor{black}{$\phantom{\frac{\frac{1}{1}}{\frac{1}{1}}}$}$\text{Price}$$_{cjt}$ & 0.4924 & 0.1803 & 0.1094 & 1.3587\tabularnewline
\textcolor{black}{$\phantom{\frac{\frac{1}{1}}{\frac{1}{1}}}$}$\text{Promo}_{cjt}$ & 0.0282 & 0.0377 & 0.0000 & 0.5000\tabularnewline
\textcolor{black}{$\phantom{\frac{\frac{1}{1}}{\frac{1}{1}}}$}$\text{Price}{}_{cjt}$
$\times$ $\text{Promo}_{cjt}$ & 0.0136 & 0.0203 & 0.0000 & 0.4505\tabularnewline
\bottomrule
\end{tabular}
\end{table}

Since the data are at the DMA level, while our approach was originally
developed for individual-level data, we now describe how to adapt
the method to the DMA-level setting. We treat the observed DMA-level
market shares $s_{cjt}$ as noisy measurements\footnote{Alternatively, with market-level data, one could treat the observed
$s_{cjt}$ as a sufficiently good approximation of $\E\left[\rest{y_{cjt}}{\bf X}_{ct},{\bf A}_{c}\right]$
as in Section 6.1 of \citet*{shi2017estimating}, in which case our
first-stage nonparametric regression is \emph{no longer} required.
We did not pursue this approach for two reasons: First, we do not
wish to impose the assumption that the observed market shares are
measured with negligible errors. Second, we intend this empirical
exercise as an illustration of our two-stage procedure, and thus focus
on a setting where the first-stage nonparametric regression is required.} of $\E\left[\rest{y_{cjt}}{\bf X}_{ct},{\bf A}_{c}\right]$, i.e.,
\[
s_{cjt}=\E\left[\rest{y_{cjt}}{\bf X}_{ct},{\bf A}_{c}\right]+u_{cjt},\quad\text{with }\E\left[\rest{u_{c}}{\bf X}_{c},{\bf A}_{c}\right]=0.
\]
Then, we use $s_{cjt}$ to nonparametrically estimate the following
intertemporal difference:
\[
\E\left[\rest{s_{cjt}-s_{cjs}}\X_{c,ts}\right]=\int\left(\E\left[\rest{y_{cjt}}{\bf X}_{ct},{\bf A}_{c}\right]-\E\left[\rest{y_{cjs}}{\bf X}_{cs},{\bf A}_{c}\right]\right)d\P\left(\rest{{\bf A}_{c}}\X_{c,ts}\right).
\]
Specifically, we nonparametrically regress $\left(s_{cjt}-s_{cjs}\right)$
on the second-order polynomial basis functions of $\mathbf{X}_{c,ts}$
with $\ell_{1}$-regularization and 10-fold cross-validation to obtain
an estimator $\hat{\g}_{j}$ of $\g_{j}\left(\overline{\X},\underline{\X}\right):=\E\left[\rest{s_{cjt}-s_{cjs}}\X_{c,ts}=\left(\overline{\X},\underline{\X}\right)\right]$.
Finally, we plug $\hat{\g}$ into our second-stage algorithm and compute
the (approximate) argmin set $\hat{B}_{\hat{c}}$.

\subsection{Results and Discussion\label{subsec:emp-Results-and-Discussion}}

We report our estimation results in Table \ref{tab:Comp}. $\hat{\b}_{\hat{c}}^{m}:=\frac{1}{2}\left(\hat{\b}_{\hat{c}}^{l}+\hat{\b}_{\hat{c}}^{u}\right)$
corresponds to the middle point of the (approximate) argmin set $\hat{B}_{\hat{c}}$
using our method. We show both the exact argmin set ($\hat{c}=0$)
and the approximate argmin set with $\hat{c}=0.1\times N^{-\frac{1}{4}}\log\left(N\right)\approx0.14$
for $N=205$. The estimated coefficients for Price (negative) and
Promo (positive) are economically intuitive.

The most interesting result is the positive estimated coefficient
on the interaction term $\text{Price}{}_{cjt}$ $\times$ $\text{Promo}_{cjt}$.
An intuitive explanation for the positive sign is that by displaying
certain products in front rows, consumers no longer see their price
tags adjacent to those of their competitors, and thus become less
price-sensitive for these specially promoted products.

\begin{table}
\caption{Empirical Illustration: Comparison of Results\label{tab:Comp}}

\bigskip{}

\noindent \centering{}%
\begin{tabular}{cccccccc}
\toprule 
\multirow{2}{*}{\textcolor{black}{$\phantom{\frac{\frac{1}{1}}{\frac{1}{1}}}$}} & \multirow{2}{*}{$\hat{\b}_{\hat{c}=0}^{m}$} & \multirow{2}{*}{$\hat{\b}_{\hat{c}=0.14}^{m}$} & \multirow{2}{*}{$\hat{\b}^{CyclicMono}$} & \multirow{2}{*}{$\hat{\b}^{OLS}$} & \multirow{2}{*}{$\hat{\b}^{OLS-FE}$} & \multirow{2}{*}{$\hat{\b}^{MLogit-FE}$} & \multirow{2}{*}{$\hat{\b}^{RCLM}$}\tabularnewline
 &  &  &  &  &  &  & \tabularnewline
\midrule
\textcolor{black}{$\phantom{\frac{\frac{1}{1}}{\frac{1}{1}}}$}$\text{Price}$$_{cjt}$ & -0.9351 & -0.9283 & $\text{-}0.3781$ & $0.0240$ & $\text{-}0.3807$ & $\text{-}0.6249$ & -0.9705\tabularnewline
\textcolor{black}{$\phantom{\frac{\frac{1}{1}}{\frac{1}{1}}}$}$\text{Promo}_{cjt}$ & 0.1793 & 0.1912 & $\text{-}0.0567$ & $0.5760$ & $0.5976$ & $0.5881$ & 0.2157\tabularnewline
\textcolor{black}{$\phantom{\frac{\frac{1}{1}}{\frac{1}{1}}}$}$\text{Price}{}_{cjt}$
$\times$ $\text{Promo}_{cjt}$ & 0.2687 & 0.2505 & $0.9240$ & $\text{-}0.8171$ & $\text{-}0.7057$ & $\text{-}0.5135$ & 0.1078\tabularnewline
\bottomrule
\end{tabular}
\end{table}

Furthermore, we compare our $\hat{\b}^{m}$ with the estimates obtained
through four other methods, i.e., Cyclic Monotonicity (CM) based on
\citet*{shi2017estimating}\footnote{We use 2-week cycles for all available weeks in the data for the CM
method.}, OLS, OLS with scalar-valued fixed effects (OLS-FE), the multinomial
logit with fixed effects (MLogit-FE), and the random coefficients
logit model (RCLM)\footnote{See Appendix \ref{subsec:Blue-Bus/Red-Bus} for the details of the
RCLM estimator.}. Results (normalized to $\S^{D-1}$) are summarized in Table \ref{tab:Comp}.

The OLS estimator for $\text{Price}$ is a positive 0.0240, which
is counterintuitive. Moreover, displaying the product in the front
rows of the store likely makes consumers less price-sensitive, suggesting
a positive coefficient on Price$\times$Promo. However, the estimated
coefficients for the interaction term using OLS, OLS-FE, and MLogit-FE
are all negative. Next, the CM-based estimator for the coefficient
of Promo is negative at -0.0567, whereas the estimated coefficient
on $\text{Price}\times\text{Promo}$ is a large positive 0.9240. While
the aggregate effect of Promo is likely to be positive for most prices
observed in the data, it makes the coefficient of Price positive for
those promoted products (i.e., $\widehat{\b}_{Price}+\text{Promo}\times\widehat{\b}_{Price\times Promo}>0$
when Promo = 1). Finally, the estimates from RCLM have the same sign
as our method. Nonetheless, it reports a smaller estimated coefficient
for the interaction term $\text{Price}{}_{cjt}\times\text{Promo}_{cjt}$,
making the effect from Promo on alleviating price sensitivity less
significant.

We view the contrast between our findings and those from alternative
methods as empirical evidence that, by accommodating more flexible
forms of unobserved heterogeneity---via high-dimensional fixed effects
that enter consumers\textquoteright{} utility functions in an additively
nonseparable manner---our approach yields more economically plausible
results.

\subsection{A Possible Explanation via Monte-Carlo Simulations}

In this subsection, we provide a possible explanation for the empirical
findings reported in Table \ref{tab:Comp} through simulation analysis.
Recall that ``Promo'' indicates whether a product receives increased
in-store exposure by being highlighted by the store. We argue that
the negative estimates on $\text{Price}{}_{ijt}$ $\times$ $\text{Promo}_{ijt}$
reported for traditional methods in Table \ref{tab:Comp} likely arise
from a positive correlation between display promotions and an unobserved
index of price sensitivity.

Specifically, suppose the utility function is
\begin{equation}
u_{ijt}=A_{ij}\times\left(X_{ijt}^{'}\b_{0}\right)+\epsilon_{ijt},\label{eq:mcemp_uijt}
\end{equation}
where $X_{ijt}$ contains Price, Promo, and Price$\times$Promo, $A_{ij}$
is the $ij$-specific fixed effect which may capture index sensitivity
(which can be thought as inversely related to unobserved brand loyalty),
and $\epsilon_{ijt}$ is the exogenous random shock. Suppose $A_{ij}$
and Promo$_{ijt}$ are positively correlated, which is reasonable
because marketing managers with their expertise are more likely to
promote products to which consumers are more price- and promotion-sensitive.
Thus, traditional estimation methods based on linearity would be unable
to detect such a pattern and wrongly attribute the effect on price
elasticities from $A_{ij}$ to Promo.

To provide some numerical evidence of the claim, we run the following
Monte Carlo simulation. We set $\b_{0}=\left(-4,2,2\right)^{'}$,
$Z_{ij}\sim\mathcal{U}\left[0,1\right]$, $A_{ij}=Z_{ij}+1$, and
$\epsilon_{ijt}\sim TIEV\left(0,1\right)$. For the $X_{ijt}$ vector,
we draw $X_{ijt,1}\sim\mathcal{U}\left[0,4\right]\text{ and }W_{ijt}\sim\mathcal{U}\left[0,1\right],$
and let $X_{ijt,2}=\left(1-\alpha\right)\times W_{ijt}+\alpha\times Z_{ij}$
and $X_{ijt,3}=X_{ijt,1}\times X_{ijt,2}$. We emphasize that $X_{ijt,2}$
(Promo) is positively correlated with $A_{ij}$ through $Z_{ij}$,
with $\alpha$ measuring the strength of the correlation. We consider
three values of $\a$: $0.15,\ 0.3,\text{ and }0.5$.

We run 1,000 simulations for each of the five methods in Table \ref{tab:Comp}
to estimate $\b_{0}$. To replicate the data structure of the empirical
exercise, we set $N=205$, $D=3$, $J=4$, and $T=10$. We report
in Table \ref{tab:MCEmp} the percentage of simulations that the corresponding
method produces correct signs for all coordinates of $X_{ijt}$.

\begin{table}
\caption{Percentage of Correct Signs of Estimated Coefficients\label{tab:MCEmp}}

\bigskip{}

\noindent \centering{}%
\begin{tabular}{>{\centering}p{1.8cm}>{\centering}p{1.8cm}>{\centering}p{1.8cm}>{\centering}p{1.8cm}>{\centering}p{1.8cm}>{\centering}p{1.8cm}>{\centering}p{1.8cm}}
\toprule 
\textcolor{black}{$\phantom{\frac{\frac{1}{1}}{\frac{1}{1}}}$}$\alpha$ & $\hat{\b}^{m}$ & $\hat{\b}^{CyclicMono}$ & $\hat{\b}^{OLS}$ & $\hat{\b}^{OLS-FE}$ & $\hat{\b}^{MLogit-FE}$ & $\hat{\beta}^{RCLM}$\tabularnewline
\midrule
\textcolor{black}{$\phantom{\frac{\frac{1}{1}}{\frac{1}{1}}}$}0.15 & 91.50\% & 0.00\% & 0.00\% & 0.00\% & 28.00\% & 0.00\%\tabularnewline
\textcolor{black}{$\phantom{\frac{\frac{1}{1}}{\frac{1}{1}}}$}0.30 & 85.90\% & 0.00\% & 0.00\% & 0.00\% & 0.20\% & 0.00\%\tabularnewline
\textcolor{black}{$\phantom{\frac{\frac{1}{1}}{\frac{1}{1}}}$}0.50 & 74.20\% & 0.00\% & 0.00\% & 0.00\% & 0.00\% & 0.00\%\tabularnewline
\bottomrule
\end{tabular}
\end{table}

The percentages of simulations where our proposed method produces
correct signs for all coordinates of $X_{ijt}$ for $\alpha=0.15,\ 0.3,$
and $0.5$ are 91.50\%, 85.90\%, and 74.20\%, respectively. The accuracy
of the estimator is negatively affected by the correlation between
$X_{ijt,2}$ (Promo) and $A_{ij}$ (multiplicative fixed effect).
In contrast, none of the other methods in Table \ref{tab:MCEmp} generates
correct signs as ours does. The alternative models, owing to their
additively separable structure,\footnote{We note that the CM method requires $A_{ij}$ entering the utility
function linearly, which is violated in \eqref{eq:mcemp_uijt}.} may overlook the positive dependence between Promo and the multiplicative
fixed effect $A_{ij}$, which can bias the resulting estimates.\footnote{Notably, RCLM produces \textquotedblleft wrong signs\textquotedblright{}
in this simulation exercise, even though it yields the expected signs
in the empirical application in Section \ref{subsec:emp-Results-and-Discussion}.
A plausible interpretation is that, while RCLM is more flexible than,
for example, MLogit-FE, the unknown selection effect in this dataset
may be insufficiently strong to cause RCLM to fail, yet strong enough
for MLogit-FE and related methods to do so. This observation highlights
the potential value of our method as a robustness-check tool.}

Intuitively, since products with larger $A_{ij}$ are more likely
to be promoted $\left(X_{ijt,2}=1\right)$ by the selection of marketing
managers, the average effective price sensitivity of promoted products
tends to be greater in magnitude than that of non-promoted products.
This drives those estimators that ignore such selection effects to
produce a negative coefficient on the interaction term. In contrast,
our method handles such \emph{non-additive} dependence between observable
characteristics and unobserved fixed effects well, illustrating its
robustness in these models.

\section{\label{sec:Conclusion}Conclusion}

\noindent This paper develops a method for semiparametric identification
and estimation in panel multinomial choice models that feature infinite-dimensional
fixed effects and nonadditive utility, thereby accommodating rich
forms of unobserved heterogeneity. We also introduce a general identification
strategy based on multivariate monotonicity of parametric indices,
applicable to a broad class of econometric models defined by the MISC
conditions. In addition, we present a computational algorithm that
leverages angle-space reparameterization and adaptive-grid search,
which prove effective given the nonstandard criterion function implied
by our identifying restrictions.

Future research could investigate how our approach might be applied
and adapted to other specific microeconometric models within the MISC
framework. Along this line, \citet*{gao2023logical}---a companion
paper to the present study---illustrates how the approach proposed
here can be adapted to the context of dyadic network formation models.
\citet{gao2023identification} propose a method for addressing endogeneity
in discrete choice models under an adapted time-homogeneity condition.
In ongoing work, we are investigating how to combine techniques developed
in this line of work to analyze strategic network formation models
with endogenous covariates.

\bibliographystyle{ecta}
\phantomsection\addcontentsline{toc}{section}{\refname}\bibliography{GL_PMC1}

\newpage{}

\appendix
\begin{singlespace}
\noindent \begin{center}
\textbf{\Large{}Online Supplemental Material for:}{\Large\par}
\par\end{center}
\end{singlespace}

\begin{center}
\textbf{\tiny{}~}{\tiny\par}
\par\end{center}

\begin{center}
\textbf{\Large{}Identification of Semiparametric Panel Multinomial}\\
\textbf{\Large{}Choice Models with Infinite-Dimensional Fixed Effects}{\Large\par}
\par\end{center}

\begin{center}
\textbf{\tiny{}~}{\tiny\par}
\par\end{center}

\setcounter{page}{1}

\section{\label{sec:Pf_sharp}Proof of Theorem \ref{thm:Id_sharp}}
\begin{proof}
Let $\b^{*}\in B_{0}\backslash\left\{ \b_{0}\right\} $. In the following,
we condition on ${\bf x}_{st}$ and suppress ${\bf x}_{st}$ for notational
simplicity. Write $p_{j\left(t\right)}:=\P\left(\rest{y_{ijt}=1}{\bf x}_{st}\right)$.
Let ${\cal J}_{+}^{*}:=\left\{ j\in{\cal J}:\left(x_{s}-x_{t}\right)^{'}\b^{*}\geq0\right\} $
and ${\cal J}_{-}^{*}:={\cal J}\backslash{\cal J}_{+}^{*}$. Then,
we have 
\begin{equation}
\sum_{j\in{\cal J}_{+}}p_{j\left(s\right)}\geq\sum_{j\in{\cal J}_{+}}p_{j\left(t\right)}.\label{eq:Delta_p_J+}
\end{equation}
Since $\sum_{j}p_{j\left(s\right)}=\sum_{j}p_{j\left(t\right)}=1$,
we also have 
\begin{equation}
\sum_{j\notin{\cal J}_{+}}p_{j\left(s\right)}\leq\sum_{j\notin{\cal J}_{+}}p_{j\left(t\right)}.\label{eq:Delta_p_J-}
\end{equation}

Without loss of generality, relabel products so that 
\[
{\cal J}_{+}^{*}=\left\{ J,\ldots,j^{*}\right\} ,\quad{\cal J}_{-}^{*}=\left\{ j^{*}-1,\ldots,1\right\} ,
\]
and
\begin{equation}
p_{J\left(s\right)}-p_{J\left(t\right)}\geq\ldots\geq p_{j^{*}\left(s\right)}-p_{j^{*}\left(t\right)},\label{eq:J+_label}
\end{equation}
and 
\begin{equation}
p_{j^{*}-1,\left(s\right)}-p_{j^{*}-1,\left(t\right)}\geq\ldots\geq p_{1,\left(s\right)}-p_{1,\left(t\right)}.\label{eq:J-label}
\end{equation}
In words, products are relabeled according to the following lexicographic
ascending order:
\begin{itemize}
\item Products in ${\cal J}_{+}$ receives larger labels than products in
${\cal J}_{-}$;
\item Within each of ${\cal J}_{+}$ and ${\cal J}_{-}$, products are further
sorted so that $p_{js}-p_{jt}$ is ascending in (new) product label
$j$.
\end{itemize}
Given \eqref{eq:Delta_p_J+} and \eqref{eq:J+_label}, we must have
\begin{equation}
\sum_{j=h}^{J}\left(p_{j\left(s\right)}-p_{j\left(t\right)}\right)\geq0\quad\forall h\geq j^{*}.\label{eq:pst_>0_k1}
\end{equation}
In the meanwhile, given \eqref{eq:Delta_p_J-} and \eqref{eq:J-label},
we also have
\[
\sum_{j=1}^{h}\left(p_{j\left(s\right)}-p_{j\left(t\right)}\right)\leq0\quad\forall h\leq j^{*}-1.
\]
Again, since $\sum_{j}p_{j\left(s\right)}=\sum_{j}p_{j\left(t\right)}=1$,
the above analysis implies that
\[
\left(1-\sum_{j=1}^{h}p_{j\left(s\right)}\right)-\left(1-\sum_{j=1}^{h}p_{j\left(t\right)}\right)\geq0\quad\forall h\leq j^{*}-1,
\]
which is equivalent to
\[
\sum_{j=h+1}^{J}\left(p_{j\left(s\right)}-p_{j\left(t\right)}\right)\geq0,\quad\forall h\leq j^{*}-1,
\]
which is further equivalent to
\begin{equation}
\sum_{j=h}^{J}\left(p_{j\left(s\right)}-p_{j\left(t\right)}\right)\geq0,\quad\forall h=1,\ldots,j^{*}.\label{eq:pst_>0_k2-1}
\end{equation}
Combining \eqref{eq:pst_>0_k1} and \eqref{eq:pst_>0_k2-1} gives
\begin{equation}
\sum_{j=h}^{J}\left(p_{j\left(s\right)}-p_{j\left(t\right)}\right)\geq0\quad\forall h=1,\ldots,J.\label{eq:pst_>0}
\end{equation}

Following \citet{pakes2016moment}, define the choice mapping $y\left({\bf x}_{s},a,\e_{s},\b\right)$
to the product label that is chosen according to model \eqref{eq:Model_PMC}
at period $s$ with $\left({\bf x}_{s},a,\e_{s},\b\right)$, i.e.,
\[
y\left({\bf x}_{s},a,\e_{s},\b\right)=j\quad\iff\quad u\left(x_{js}^{'}\b,a_{j},\e_{js}\right)>\max_{k\neq j}u\left(x_{ks}^{'}\b,a_{k},\e_{ks}\right).
\]
Take $a$ to be any in-support value in ${\cal A}$. Define
\[
R_{j;s}^{*}:=\left\{ \tilde{\e}_{s}:y\left({\bf x}_{s},a,\tilde{\e}_{s},\b^{*}\right)=j\right\} 
\]
to be the values of $\e_{s}^{*}$ that lead to product $j$ being
chosen under $\left({\bf x}_{s},a,\b^{*}\right)$ in period $s$.

Recall that $\left(x_{js}-x_{jt}\right)^{'}\b^{*}\geq0$ for $j\in{\cal J}_{+}^{*}$
and $\left(x_{js}-x_{jt}\right)^{'}\b^{*}<0$ for $j\in{\cal J}\backslash{\cal J}_{+}^{*}$.
Hence, by the monotonicity of $u$ in its first argument,
\begin{equation}
\bigcup_{j\in{\cal J}_{+}^{*}}R_{j;t}^{*}\subseteq\bigcup_{j\in{\cal J}_{+}^{*}}R_{j;s}^{*},\label{eq:SetInclusion-1}
\end{equation}
and hence
\[
R_{j;s}^{*}\cap R_{k;t}^{*}=\es\quad\forall j\in{\cal J}_{-}^{*},\ k\in{\cal J}_{+}^{*}.
\]

Let $R_{j,k}^{*}:=R_{j;s}^{*}\cap R_{k;t}^{*}$ and 
\[
R_{j,k,h,l}^{*}:=R_{j,k}^{*}\times R_{h,l}^{*}:=\left\{ \left(\tilde{\e}_{s},\tilde{\e}_{t}\right):\tilde{\e}_{s}\in R_{j,k}^{*},\tilde{\e}_{t}\in R_{h,l}^{*}\right\} .
\]
Whenever $R_{j,k}^{*}\neq\es$, pick any single point $r_{j,k}^{*}\in R_{j,k}^{*}$.

We specify the joint distribution of $\left(\e_{s}^{*},\e_{t}^{*}\right)$
as a discrete distribution over points $\left(r_{j,k}^{*},r_{h,l}^{*}\right)$.
We write 
\[
q_{j,k,h,l}^{*}:=\P\left(\rest{\left(\e_{s}^{*},\e_{t}^{*}\right)=\left(r_{j,k}^{*},r_{h,l}^{*}\right)}{\bf x}_{s},{\bf x}_{t}\right)\equiv\P\left(\rest{\left(\e_{s}^{*},\e_{t}^{*}\right)\in R_{j,k,h,l}^{*}}{\bf x}_{s},{\bf x}_{t}\right).
\]
Hence, a vector $\ol q^{*}=\left(q_{j,k,h,l}^{*}\right)$ in the unit
simplex defines a joint distribution of $\left(\e_{s}^{*},\e_{t}^{*}\right)$
.

By \eqref{eq:SetInclusion-1}, we have
\begin{equation}
q_{j,k,h,l}^{*}=0,\ \forall\left(j,k\right)\text{ or }\left(h,l\right)\in{\cal J}_{-}^{*}\times{\cal J}_{+}^{*}.\label{eq:q_-+_0}
\end{equation}
We set 
\begin{equation}
q_{j,k,h,l}^{*}=0,\ \forall j<k\text{ or }h<l.\label{eq:q_jkhl_0}
\end{equation}
Note that $\left(j,k\right)\in{\cal J}_{-}^{*}\times{\cal J}_{+}^{*}$
implies $j<k$ by the relabeling. Hence, \eqref{eq:q_jkhl_0} sets
a larger class of $q_{jkhl}$ to be zero than those in \eqref{eq:q_-+_0}.

Given \eqref{eq:q_jkhl_0}, to specify a joint distribution of $\left(\e_{s}^{*},\e_{t}^{*}\right)$,
we only need to specify 
\begin{equation}
q^{*}:=\left(q_{j,k,h,l}^{*}\right)_{j\geq k,h\geq l},\label{eq:q*}
\end{equation}
so that $\ol q^{*}:=\left(q^{*},\ul q^{*}\right)$ with 
\[
\ul q^{*}:=\left(q_{j,k,h,l}^{*}\right)_{j<k\text{ or }h<l}={\bf 0}.
\]

First, $q^{*}$ needs to satisfy a homogeneity condition. Note that,
given \eqref{eq:q_jkhl_0},
\begin{align*}
\P\left(\rest{\e_{s}^{*}=r_{jk}}{\bf x}_{s},{\bf x}_{t}\right) & =\sum_{h,l}q_{j,k,h,l}^{*}=\sum_{h\geq l}q_{j,k,h,l}^{*}\\
\P\left(\rest{\e_{t}^{*}=r_{jk}}{\bf x}_{s},{\bf x}_{t}\right) & =\sum_{h,l}q_{h,l,j,k}^{*}=\sum_{h\geq l}q_{h,l,j,k}^{*}
\end{align*}
For $j<k$, conditional homogeneity is trivially satisfied, since
\[
\P\left(\rest{\e_{s}^{*}=r_{jk}}{\bf x}_{s},{\bf x}_{t}\right)=\P\left(\rest{\e_{t}^{*}=r_{jk}}{\bf x}_{s},{\bf x}_{t}\right)=0\text{ }\forall j<k.
\]
Hence, to satisfy homogeneity condition, we just need to ensure 
\begin{equation}
\sum_{h\geq l}q_{j,k,h,l}^{*}=\sum_{h\geq l}q_{h,l,j,k}^{*},\quad\forall j\geq k.\label{eq:eq_CH}
\end{equation}

Second, $q^{*}$ needs to produce choice probabilities that match
the true probabilities. Let 
\begin{align*}
p_{jk}^{*} & :=\P\left(\rest{y\left({\bf x}_{s},a,\e_{s}^{*},\b\right)=j,y\left({\bf x}_{t},a,\e_{t}^{*},\b\right)=k}{\bf x}_{s},{\bf x}_{t}\right)\\
 & =\P\left(\rest{\left(\e_{s}^{*},\e_{t}^{*}\right)\in R_{j;s}^{*}\times R_{k;t}^{*}}{\bf x}_{s},{\bf x}_{t}\right).
\end{align*}
Then, the joint choice probabilities
\[
p_{jk}^{*}=\sum_{h=1}^{J}\sum_{l=1}^{J}q_{j,h,l,k}^{*}=\sum_{h\leq j}\sum_{l\geq k}q_{j,h,l,k}^{*}
\]
To match with true probabilities, we need 
\begin{equation}
\sum_{h\leq j}\sum_{l\geq k}q_{j,h,l,k}^{*}=p_{jk},\forall j,k=1,\ldots,J\label{eq:eq_TP}
\end{equation}

Given $p_{j\left(s\right)}=\sum_{k}p_{jk}$ and $p_{j\left(t\right)}=\sum_{k}p_{kj}$,
we can now translate condition \eqref{eq:pst_>0} about $p_{j\left(s\right)}-p_{j\left(t\right)}$
into a restriction about $\left(p_{jk}\right)$:
\begin{equation}
\sum_{j=h}^{J}\sum_{k}p_{jk}\geq\sum_{j=h}^{J}\sum_{k}p_{kj},\quad\forall h=1,\ldots,J.\label{eq:cond_pjk}
\end{equation}

To summarize, given \eqref{eq:pst_>0}, a condition on $p_{jk}$ that
we know to hold, we need to show that there exists a $q^{*}:=\left(q_{j,k,h,l}^{*}\right)_{j\geq k,h\geq l}$
such that \eqref{eq:eq_CH} and \eqref{eq:eq_TP} both hold. 

The above, however, is exactly the same as the one solved in \citet*{pakes2016moment},
as summarized by the following lemma. 
\end{proof}
\begin{lem}[Nonnegative Solvability of Linear System in \citet{pakes2016moment}]
\label{lem:PPsharp}Let $q^{*}\equiv\left(q_{jkhl}^{*}\right)$ be
a $\left(\frac{1}{2}J\left(J+1\right)\right)^{2}$-dimensional vector
indexed by $j,k,h,l\in\left\{ 1,\ldots,J\right\} $ such that
\[
j\geq k\quad\text{and}\quad h\geq l.
\]
Let $p\equiv\left(p_{jk}\right)$ be a $J^{2}$-dimensional vector
indexed by $j,k\in\left\{ 1,\ldots,J\right\} $ s.t. $p\in\D^{J^{2}-1}$,
i.e., $p$ is a (discrete) probability distribution over $J^{2}$
points defined by $jk$.

Suppose that
\begin{equation}
\sum_{j=h}^{J}\sum_{k}p_{jk}\geq\sum_{j=h}^{J}\sum_{k}p_{kj},\quad\forall h=1,\ldots,J.\label{eq:st_ineq}
\end{equation}
Then, there exists a nonnegative $q^{*}\geq0$ that solves the following
joint system of equations:
\begin{align}
\sum_{h\leq j}\sum_{l\geq k}q_{j,h,l,k}^{*}=p_{jk}, & \quad\forall j,k\in\left\{ 1,\ldots,J\right\} .\label{eq:Eq_Prob}\\
\sum_{h}\sum_{l\leq h}\left(q_{h,l,j,k}^{*}-q_{j,k,h,l}^{*}\right)=0, & \quad\forall j,k\in\left\{ 1,\ldots,J\right\} \text{ s.t. }j\geq k.\label{eq:Eq_Stat}
\end{align}
\end{lem}
Lemma \ref{lem:PPsharp} above summarizes the part of \citet{pakes2016moment}'s
proof of their Theorem 2 that is useful to our setup. Specifically,
the inequality constraints in \eqref{eq:st_ineq} are the same as
those in equation (S5) of \citet{pakes2016moment}; equations in \eqref{eq:Eq_Prob}
are the same as those in (S3) of \citet{pakes2016moment}; equations
in \eqref{eq:Eq_Stat} are the same as those in (S4) of \citet{pakes2016moment}.
The proof in \citet{pakes2016moment} after their equation (S5) then
establishes the conclusion of Lemma \ref{lem:PPsharp}.

\section{\label{sec:App_pf_conv}Proof of Theorem \ref{thm:Consistency}}

We first prove two lemmas before formally proving Theorem \ref{thm:Consistency}.
\begin{lem}
\label{lem:Q_Cts} $Q:\S^{d-1}\to\R_{+}$ is continuous.
\end{lem}
\begin{proof}
Recalling that $v_{k}\left(\ol{\X}-\ul{\X}\right)=\ol X_{k}-\ul X_{k}/\norm{\ol X_{k}-\ul X_{k}}$
whenever $\ol X_{k}\neq\ul X_{k}$ while $v_{k}\left(\ol{\X}-\ul{\X}\right)=0$
when $\ol X_{k}=\ul X_{k}$, we have
\begin{align*}
G\left(\g_{\tJ,t,s}\left(\X_{i,ts}\right)\right)\l_{j}\left(\X_{i,ts};\b\right)= & \ G\left(\g_{\tJ,t,s}\left(\X_{i,ts}\right)\right)\prod_{k=1}^{J}\ind\left\{ \left(-1\right)^{\ind\left\{ k\in\tJ\right\} }\left(X_{ikt}-X_{iks}\right)^{'}\b\geq0\right\} \\
= & \ G\left(\g_{\tJ,t,s}\left(\X_{i,ts}\right)\right)\prod_{k=1}^{J}\ind\left\{ \left(-1\right)^{\ind\left\{ k\in\tJ\right\} }v_{k}\left({\bf X}_{it}-\X_{is}\right)^{'}\b\geq0\right\} ,
\end{align*}
which is continuous in $\b$ with probability one, since $v_{k}\left({\bf X}_{it}-\X_{is}\right)$
has no mass point except possibly at ${\bf 0}$, in which case the
indicator degenerates to a constant over $\b\in\S^{d-1}$. Since $\X_{i,ts}$
is i.i.d. across $i$, $\S^{d-1}$ is compact, and the indicator function
is bounded, all conditions for Lemma 2.4 in \citet{newey1994asymp}
are satisfied, by which we conclude that $Q=\sum_{\tJ,t,s}Q_{\tJ,t,s}$
is continuous on $\S^{d-1}$.
\end{proof}
\begin{lem}
\label{lem:Qn_conv}Under Assumptions \ref{assu:RandSamp}, \ref{assu:NiceG},
and \ref{assu:NoMass}, $\sup_{\b\in\S^{d-1}}\left|\hat{Q}\left(\b\right)-Q\left(\b\right)\right|=O_{p}\left(c_{N}\right).$
\end{lem}
\begin{proof}
We first prove the convergence of $\hat{Q}_{\tJ,t,s}\left(\b\right)$
to $Q_{\tJ,t,s}\left(\b\right)$ for each $\left(\tJ,t,s\right)$.
For each generic deterministic function $\tilde{\g}_{\tJ,t,s}$, define
\begin{align*}
Q_{\tJ,t,s}\left(\b,\tilde{\g}\right) & :=\E\left[G\left(\tilde{\g}_{\tJ,t,s}\left(\X_{i,ts}\right)\right)\l_{\tJ}\left(\X_{i,ts};\b\right)\right],\\
\hat{Q}_{\tJ,t,s}\left(\b,\tilde{\g}\right) & :=\frac{1}{n}\sum_{i=1}^{n}G\left(\tilde{\g}_{\tJ,t,s}\left(\X_{i,ts}\right)\right)\l_{\tJ}\left(\X_{i,ts};\b\right),
\end{align*}
so that $\hat{Q}_{\tJ,t,s}\left(\b\right)=\hat{Q}_{\tJ,t,s}\left(\b,\tilde{\g}_{\tJ,t,s}\right)$
and $Q_{\tJ,t,s}\left(\b\right)=Q_{\tJ,t,s}\left(\b,\g\right)$. For
notational simplicity, we suppress the subscript $\left(\tJ,t,s\right)$
for the moment.

Defining ${\cal Q}:=\left\{ G\left(\tilde{\g}\left({\bf x}_{ts}\right)\right)\l\left({\bf x}_{ts};\b\right):\tilde{\g}\in\G,\b\in\S^{d-1}\right\} $,
we first argue that ${\cal Q}$ is a $\P$-Donsker class based on
\citet*{van1996weak}. First, it is easy to show by Assumption \ref{assu:NiceG}
that $G\left(0\right)=0$, which together with the Lipschitz continuity
of $G$, implies that $\E\left[G^{2}\left(\tilde{\g}\left({\bf X}_{i}\right)\right)\right]\leq M\E\left[\tilde{\g}^{2}\left({\bf X}_{i}\right)\right]<\infty$
and $\E\left|G\left(\tilde{\g}\left({\bf X}_{i}\right)\right)\right|\leq\E\left|\tilde{\g}\left({\bf X}_{i}\right)\right|\leq\sup_{\tilde{\g}\in\G}\E\left|\tilde{\g}\left({\bf X}_{i}\right)\right|<\infty$.
Then, as $\G$ is $\P$-Donsker, $G\circ\tilde{\g}$ must also be
$\P$-Donsker. Second, recall that $\l\left(\X_{i,ts};\b\right)$
is the product of indicators of half planes, while the collection
of $\ind\left\{ \left({\bf x}_{kt}-{\bf x}_{ks}\right)^{'}\b\geq0\right\} $
over $\b\in\S^{d-1}$ is a well-known VC-class of functions (sets)
and is thus $\P$-Donsker. Finally, since the indicator function is
uniformly bounded and $\sup_{\tilde{\g}\in\G}\E\left|G\left(\tilde{\g}\left({\bf X}_{i}\right)\right)\right|<\infty$,
we conclude that ${\cal Q}$ is also $\P$-Donsker:
\begin{equation}
\sup_{\b\in\S^{d-1}}\sup_{\tilde{\g}\in\G}\left|\hat{Q}\left(\b,\tilde{\g}\right)-Q\left(\b,\tilde{\g}\right)\right|=O_{p}\left(N^{-\frac{1}{2}}\right).\label{eq:GC}
\end{equation}
Next, by Assumption \ref{assu:FS_Conv}, we have
\begin{align}
\sup_{\b\in\S^{d-1}}\left|Q\left(\b,\hat{\g}\right)-Q\left(\b,\g\right)\right|\leq & \sup_{\b\in\S^{d-1}}\int\left|G\left(\hat{\g}\left({\bf x}_{ts}\right)\right)-G\left(\g\left({\bf x}_{ts}\right)\right)\right|\l_{j}\left({\bf x}_{ts};\b\right)d\P\left({\bf x}_{ts}\right)\nonumber \\
\leq & \ M\sqrt{\int\left(\hat{\g}\left({\bf x}_{ts}\right)-\g\left({\bf x}_{ts}\right)\right)^{2}d\P\left({\bf x}_{ts}\right)}=O_{p}\left(c_{N}\right)\label{eq:Bias}
\end{align}
by Lipschitz continuity of $G$, $\left|\l_{j}\right|\leq1$, and
the Cauchy-Schwarz inequality. Combining \eqref{eq:GC} and \eqref{eq:Bias},
we have
\begin{align*}
\sup_{\b\in\S^{d-1}}\left|\hat{Q}\left(\b,\hat{\g}\right)-Q\left(\b,\g\right)\right|\leq & \sup_{\b\in\S^{d-1}}\left|\hat{Q}\left(\b,\hat{\g}\right)-Q\left(\b,\hat{\g}\right)\right|+\sup_{\b\in\S^{d-1}}\left|Q\left(\b,\hat{\g}\right)-Q\left(\b,\hat{\g}\right)\right|\\
\leq & \sup_{\b\in\S^{d-1}}\sup_{\tilde{\g}\in\G}\left|\hat{Q}\left(\b,\tilde{\g}\right)-Q\left(\b,\tilde{\g}\right)\right|+\sup_{\b\in\S^{d-1}}\left|Q\left(\b,\hat{\g}\right)-Q\left(\b,\hat{\g}\right)\right|\\
= & \ O_{p}\left(N^{-\frac{1}{2}}\right)+O_{p}\left(c_{N}\right)=O_{p}\left(c_{N}\right)
\end{align*}
since $N^{-\frac{1}{2}}=O_{p}\left(c_{N}\right)$ for nonparametric
estimators. Summing over all $\left(j,t,s\right)$, we have $\sup_{\b\in\S^{d-1}}\left|\hat{Q}\left(\b\right)-Q\left(\b\right)\right|=O_{p}\left(c_{N}\right)$.
\end{proof}

\subsubsection*{Main Proof of Theorem \ref{thm:Consistency}}
\begin{proof}
We verify Condition C.1 in \citet*[CHT thereafter]{chernozhukov2007estimation}
so as to apply their Theorem 3.1. Condition C.1(a) on the non-emptiness
and compactness of parameter space is satisfied given Theorem \ref{thm:SetID}.
Condition C.1(b) on the continuity of the population criterion function
$Q$ is proved by Lemma \ref{lem:Q_Cts}. Condition C.1(c) on measurability
of the sample criterion function is satisfied by construction. Conditions
C.1(d)(e) regarding the uniform convergence of $Q_{n}$ are satisfied
by Lemma \ref{lem:Qn_conv}. Hence, Theorem 3.1.(1) in CHT implies
the Hausdorff consistency of $\hat{B}$. The consistency of the point
estimator under the additional assumption of point identification
(i.e., $B_{0}$ is a singleton) follows from Theorem 3.2 of CHT.
\end{proof}

\section{\label{sec:App_PID}Sufficient Conditions for Point Identification}

In this section, we prove sufficient conditions for the point identification
of $\b_{0}$. For simplicity of notation, we fix $T=2$, and focus
on the conditions that would establish point identification using
identifying restrictions from singleton products $j=1,\ldots,J$.
Consequently, these sufficient conditions should be regarded as being
strictly more than necessary.

Define $\bs{\d}_{t}:=\left(\d_{1t},\ldots,\d_{Jt}\right)^{'}$, where
$\d_{jt}:=x_{jt}^{'}\b_{0}$, and recall that $\d_{ijt}=X_{ijt}^{'}\b_{0}$.
Denote
\[
\psi_{j}\left(\bs{\d}_{t};{\bf x}_{ts},{\bf A}_{i}\right):=\int\ind\left\{ u\left(\d_{jt},A_{ij},\e_{ijt}\right)>\max_{k\neq j}u\left(\d_{kt},A_{ik},\e_{ikt}\right)\right\} \text{d}{\cal \P}\left(\rest{\bs{\e}_{it}}{\bf X}_{i,ts}={\bf x}_{ts},{\bf A}_{i}\right)
\]
We first impose a strict multivariate monotonicity condition on $\psi_{j}\left(\cd\right)$.
\begin{assumption}[Strict Monotonicity of $\psi_{j}$]
\label{assu:psiA_strictmono} For any realized ${\bf A}_{i}$ and
${\bf x}_{ts}$, the function $\psi_{j}\left(\bs{\d}_{t};{\bf x}_{ts},{\bf A}_{i}\right):\R^{J}\to\R$
is strictly increasing in $\d_{jt}$ and decreasing in $\d_{kt}$
for $k\neq j$, i.e., if for any two periods $t$ and $s$ it is true
that $\d_{jt}>\d_{js}$ and $\d_{kt}<\d_{ks}$ for all $k\neq j$,
then $\psi_{j}\left(\bs{\d}_{t};{\bf x}_{ts},{\bf A}_{i}\right)>\psi_{j}\left(\bs{\d}_{s};{\bf x}_{ts},{\bf A}_{i}\right).$
\end{assumption}
\noindent We note that Assumption \ref{assu:psiA_strictmono} is implied
by a stronger version of Assumption \ref{assu:PMC_Mono} together
with an additional condition on the support of $u$ given $\left({\bf X}_{i},{\bf A}_{i}\right)$. 

\medskip{}

\noindent \textbf{Assumption 2$'$} \customlabel{as1prime}{2$'$}(Strict
Monotonicity of $u$)\textbf{.} $u\left(\d_{ijt},A_{ij},\e_{ijt}\right)$
\textit{is strictly increasing in the index $\d_{ijt}$, for every
realization of} $\left(A_{ij},\e_{ijt}\right)$.

\medskip{}

\noindent \textbf{Assumption 2$''$} \customlabel{as1pp}{2$''$}(Overlapping
Supports)\textbf{.} \textit{Conditional on any realization of ${\bf X}_{i}$
and ${\bf A}_{i}$, we have} $\bigcap_{j=1}^{J}int\left(Supp\left(u\left(\d_{ijt},A_{ij},\e_{ijt}\right)\right)\right)\neq\es$.

\medskip{}

\noindent In particular, Assumption \ref{as1pp} is directly implied
by the assumption of $\text{Supp}\left(u\left(\d_{ijt},A_{ij},\e_{ijt}\right)\right)$
being equal to the real line conditional on any realization of ${\bf X}_{i}$
and ${\bf A}_{i}$, which is again satisfied in additive panel multinomial
choice models with scalar fixed effects a la 
\[
u\left(\d_{ijt},A_{ij},\e_{ijt}\right)=\d_{ijt}+A_{ij}+\e_{ijt}
\]
under the assumption of $\text{Supp}\left(\rest{\e_{ijt}}{\bf X}_{i},{\bf A}_{i}\right)=\R$
as commonly imposed in the literature.
\begin{lem}
Assumptions \ref{as1prime} and \ref{as1pp} imply Assumption \ref{assu:psiA_strictmono}.
\end{lem}
\noindent Finally, we impose the following assumption on $\D{\bf X}_{i}$,
with $\D X_{ij}:=X_{ij1}-X_{ij2}$ for all individual $i$ and product
$j$ between period 1 and period 2.
\begin{assumption}[Full-Directional Support of $\D{\bf X}_{i}$]
\label{assu:DeltaXSupp} Suppose either (a) or (b) is true:
\begin{itemize}
\item[(a)]  ${\bf 0}\in int\left(Supp\left(\D{\bf X}_{i}\right)\right)$.
\item[(b)]  There exists some $k\in\left\{ 1,\ldots,D\right\} $ such that $\b_{0}^{k}\neq0$
and $Supp\left(\rest{\D X_{ij}^{k}}\D X_{il},\sp l\neq j\right)=\R$
for all $j\in\left\{ 1,\ldots,J\right\} $. Furthermore, for all $j\in\left\{ 1,\ldots,J\right\} $,
$Supp\left(\rest{\D X_{ij}}\D X_{il},\sp l\neq j\right)$ is not contained
in a proper linear subspace of $\R^{D}$.
\end{itemize}
\end{assumption}
\noindent Assumption \ref{assu:DeltaXSupp}(a) is satisfied when $X_{ij}$
is a continuous random vector. On the other hand, Assumption \ref{assu:DeltaXSupp}(b)
can accommodate discrete regressors generally, but requires one continuous
covariate with a large support. Assumption \ref{assu:DeltaXSupp}
ensures that $\D X_{ij}^{'}\b_{0}>0$ and $\D X_{ik}^{'}\b_{0}<0$
for all $k\neq j$ hold simultaneously with strictly positive probability.
\begin{thm}[Point Identification]
\label{thm:consistencyPMC} Under Assumptions \ref{assu:RandSamp},
\ref{assu:EpsDist}, \ref{assu:psiA_strictmono}, and \ref{assu:DeltaXSupp},
$\b_{0}$ is point identified on $\S^{D-1}$.
\end{thm}
\begin{proof}
Fix any ${\bf x}_{ts}$ in the support of ${\bf X}_{i}$. By definition
of $\g_{j,t,s}$, we have
\begin{align*}
\g_{j,t,s}\left({\bf x}_{ts}\right) & =\int\left[\psi_{j}\left(\bs{\d}_{t};{\bf x}_{ts},{\bf A}_{i}\right)-\psi_{j}\left(\bs{\d}_{s};{\bf x}_{ts},{\bf A}_{i}\right)\right]\text{d}\P\left(\rest{{\bf A}_{i}}{\bf X}_{i,ts}={\bf x}_{ts}\right).
\end{align*}
Hence, under Assumption \ref{assu:psiA_strictmono}, we have
\begin{equation}
\d_{jt}>\d_{js}\text{ and }\d_{kt}<\d_{ks}\text{ for all }k\neq j\quad\imp\quad\g_{j,t,s}\left({\bf x}_{ts}\right)>0,\label{eq:gamma_stmono}
\end{equation}
since $\psi_{j}\left(\bs{\d}_{t};{\bf x}_{ts},{\bf A}_{i}\right)>\psi_{j}\left(\bs{\d}_{s};{\bf x}_{ts},{\bf A}_{i}\right)$
for every realization of ${\bf A}_{i}$ given ${\bf X}_{i,ts}={\bf x}_{ts}$.
Together with Assumption \ref{assu:DeltaXSupp}, we deduce that
\begin{align*}
\P\left\{ \g_{j,t,s}\left({\bf X}_{i}\right)>0\right\}  & \geq\P\left\{ \D X_{ij}^{'}\b_{0}>0\ \wedge\ \D X_{ik}^{'}\b_{0}<0,\sp\any k\neq j\right\} >0.
\end{align*}
Now for any $\b\in\S^{D-1}\backslash\left\{ \b_{0}\right\} $, define
for any product $j$:
\[
H_{j}\left(\b\right):=\left\{ {\bf v}\in Supp\left(\D\X_{i}\right):\sp v_{j}^{'}\b<0<v_{j}^{'}\b_{0},\ \wedge\ v_{k}^{'}\b_{0}<0<v_{k}^{'}\b,\,\any k\neq j\right\} .
\]
As $\b\neq\b_{0}$, by Assumption \ref{assu:DeltaXSupp} we have 
\begin{equation}
\P\left(\D\X_{i}\in H_{j}\left(\b\right)\right)>0.\label{eq:pos_prob}
\end{equation}
Moreover, for any realization of ${\bf X}_{i}$ such that $\D\X_{i}\in H_{j}\left(\b\right)$,
we must have: (i) $\g_{j,t,s}\left({\bf X}_{i}\right)>0$ by \eqref{eq:gamma_stmono},
and (ii)
\[
\l_{j}\left(\D\X_{i},\b\right)=\prod_{k=1}^{J}\ind\left\{ \left(-1\right)^{\ind\left\{ k=j\right\} }\D X_{ik}^{'}\b\geq0\right\} =1
\]
so that 
\[
G\left(\g_{j}\left({\bf X}_{i}\right)\right)\l_{j}\left(\D\X_{i},\b\right)=G\left(\g_{j}\left({\bf X}_{i}\right)\right)>0
\]
for all such ${\bf X}_{i}$. Hence, we have
\begin{equation}
\E\left[\rest{G\left(\g_{j}\left({\bf X}_{i}\right)\right)}\D\X_{i}\in H_{j}\left(\b\right)\right]>0.\label{eq:pos_condE}
\end{equation}
Combining \eqref{eq:pos_prob} and \eqref{eq:pos_condE}, we have
\begin{align*}
Q_{j}\left(\b\right) & =\E\left[G\left(\g_{j}\left({\bf X}_{i}\right)\right)\l_{j}\left(\D\X_{i},\b\right)\right]\\
 & \geq\E\left[G\left(\g_{j}\left({\bf X}_{i}\right)\right)\l_{j}\left(\D\X_{i},\b\right)\ind\left\{ \D\X_{i}\in H_{j}\left(\b\right)\right\} \right]\\
 & =\E\left[G\left(\g_{j}\left({\bf X}_{i}\right)\right)\ind\left\{ \D\X_{i}\in H_{j}\left(\b\right)\right\} \right]\\
 & =\E\left[\rest{G\left(\g_{j}\left({\bf X}_{i}\right)\right)}\D\X_{i}\in H_{j}\left(\b\right)\right]\P\left(\D\X_{i}\in H_{j}\left(\b\right)\right)\\
 & >0=Q_{j}\left(\b_{0}\right).
\end{align*}
\end{proof}

\section{More on the Adaptive-Grid Algorithm\label{sec:app_GridSearch}}

We elaborate on the adaptive-grid algorithm described in Section \ref{subsubsec:Comp}
and provide additional details regarding its practical implementation.
A graphical illustration of the algorithm is presented in Figure \ref{fig:Algo}.
\begin{enumerate}
\item \textbf{Initialization and Coarse Grid Search.} The search begins
by defining a coarse, uniform grid of $M_{0}^{D-1}$ points over the
initial parameter space $\T$. The code sets $M_{0}$ via the \texttt{M\_Step}
parameter. The objective function $Q(\beta)$ is evaluated at each
point on this grid. This stage then iteratively refines the search
space. After each evaluation, the algorithm prunes the parameter space
by retaining only the region containing the points that yielded the
lowest 20\% of objective function values. A new, equally-sized grid
is then constructed within this smaller, more promising region, and
the process repeats. This iterative pruning continues for a fixed
number of loops or until the minimum objective function value stabilizes,
effectively narrowing the search to a region likely to contain the
global minimum. The procedure includes a check to recenter the search
space if the minimum appears to lie on a boundary of the angular coordinates,
which prevents erroneously discarding a wrapped-around solution.
\item \textbf{Fine Grid Refinement.} Upon completion of the coarse search,
the algorithm enters a refinement stage. It constructs a new, finer
grid within a small neighborhood surrounding the set of candidate
points that produced the minimum objective function value in the previous
stage. The resolution of this grid is determined by a smaller step
size, \texttt{Tol\_Step}. By re-evaluating the objective function
on this denser set of points, this stage seeks to improve upon the
initial estimate. If a new, lower minimum value is found, the process
repeats around this new set of points. This stage concludes when no
further improvement in the minimum of the objective function can be
achieved, indicating that the algorithm has converged to a local minimum
within the refined region.
\item \textbf{Boundary Identification.} The final stage aims to precisely
delineate the identified set of minimizing parameters. A very fine
grid is constructed in the immediate buffer zone \emph{around} the
set of optimal points found in the prior stage, while excluding the
interior of this set. This allows the algorithm to meticulously probe
the boundary of the solution set. Any points on this boundary that
also yield the minimum objective function value are added to the solution
set. This step repeats, further reducing the grid step size until
a pre-specified tolerance (\texttt{Tol}) is met and a clear separation
exists between the parameters that minimize the function and those
that do not. This ensures the final identified set is not only optimal
but also well-defined. 
\end{enumerate}
\begin{figure}
\caption{\label{fig:Algo}An Adaptive-Grid Algorithm}

\medskip{}

\noindent \centering{}\begin{tikzpicture}[scale=4,every node/.style={circle,inner sep=0pt,minimum size=0.3cm}]

\fill[fill=red!30] (1,-0.3) rectangle (-1,0.3);

\draw (-1, -0.3) node[below right,font=\fontsize{12pt}{0}]
{$ -\pi$};

\draw (1, -0.3) node[below=7pt, left,font=\fontsize{12pt}{0}] {$ \pi$};

\draw (-1, -0.3) node[above=8pt, left,font=\fontsize{10pt}{0}] {$ -\pi/2 $};

\draw (-1, 0.3) node[below=8pt, left,font=\fontsize{10pt}{0}] {$ \pi/2 $};


\draw (0, -0.3) node[below=2pt,font=\fontsize{10pt}{0}] {$ 0 $};

\draw [dashed, red!30] (1, -0.3) -- (2,-0.3);
\draw [dashed, red!30] (1, 0.3) -- (2,0.3);
\draw (2, -0.3) node[below,font=\fontsize{10pt}{0}] {$ 2\pi $};
\fill[fill=red!10] (1,0.3) rectangle (2,-0.3);

\fill[fill=blue!30] (-1,0.2) rectangle (-0.8,-0.2);

\fill[fill=blue!30] (0.2,0.2) rectangle (1,-0.2);

\fill[fill=blue!10] (1,0.2) rectangle (1.2,-0.2);

\fill[fill=blue!50] (0.5,0.1) rectangle (0.8,-0.1);

\draw (0.65, 0) node[font=\fontsize{8pt}{0}] {$ \Theta_0 $};
\end{tikzpicture}
\end{figure}
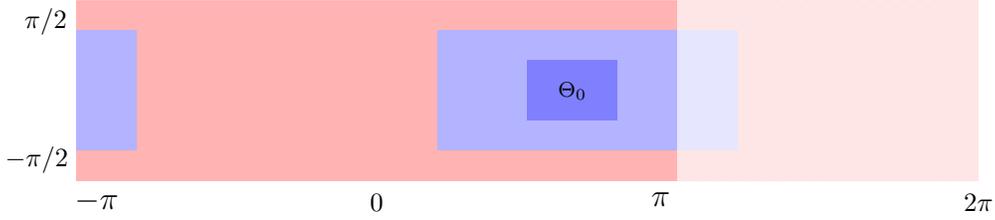

Figure \ref{fig:Algo} provides a graphical illustration of our adaptive-grid
algorithm to find the minimizer of the criterion function over the
parameter space $\T$. The figure corresponds to a re-parameterization
in the 2-dimensional angle space of a 3-dimensional unit ball, as
specified in equation (\ref{eq:Theta_hat_CHT07}) of the paper for
$D=3$. The pink area represents the initial angle space. The light
blue rectangle corresponds to $\underline{\Theta}^{(1)}$ defined
in (\ref{eq:Theta_ul1}). The dark blue rectangle represents the identified
set $\Theta_{0}$, which is the best we can hope for. The disconnectedness
of the light blue rectangle corresponds to Step 3 of the algorithm
presented in Section \ref{subsubsec:Comp}. This is largely a programming
nuisance: it means that if $\underline{\Theta}^{(1)}\not\subset\Theta^{(1)}$
due to the special structure of the angle space, we add $2\pi$ to
$\underline{\Theta}^{(1)}$ such that it is connected and lies within
$\Theta^{(1)}$. Note that adding $2\pi$ to the first coordinate
of $\t$ does not change the value of $\theta$ nor $\b$. 

\section{Counterfactual Analysis in Long Panels \label{sec:Counterfactual-Analysis}}

So far, we have focused on the identification and estimation of the
index parameter $\b_{0}$. While $\b_{0}$ itself contains rich information
about the unobserved preference of consumers, we are often also interested
in counterfactual parameters defined as some functional of not only
$\b_{0}$, but also other unknown components of the model. In this
section, we discuss how the estimate $\hat{\b}$ of $\b_{0}$, and
the computed indices based on $\hat{\b}$, can be used to estimate
more sophisticated counterfactual parameters.

\textcolor{black}{An important class of counterfactual parameters
concerns the prediction of counterfactual market shares (aggregate
choice probability) in the form of
\begin{equation}
\mu_{j}\left(\ol{{\bf X}}\right):=\int\E\left[\rest{y_{ijt}}{\bf X}_{it}=\ol{{\bf X}},{\bf A}_{i}\right]\text{d}\P\left({\bf A}_{i}\right),\label{eq:cf_mu_j}
\end{equation}
Demand elasticities may be further computed as $\Dif\mu_{j}\left(\ol{{\bf X}}\right)$,
which gives the marginal effect of an exogenous change in certain
observable characteristics on consumer choices.}\footnote{\textcolor{black}{It is important to note that in the expression of
$\mu$ we use the marginal distribution $\P\left({\bf A}_{i}\right)$
rather than the conditional distribution $\P\left(\rest{{\bf A}_{i}}{\bf X}_{it}=\ol{\X}\right)$.
This separation between} the exogenously imposed counterfactual $\ol{\X}$
and \textcolor{black}{the distribution of the unobserved ${\bf A}_{i}$
}is key to the interpretation of $\mu_{j}\left(\ol{\X}\right)$ as
the direct effect of \textcolor{black}{the exogenous change} in observable
characteristics $\X$ on choice probabilities, with the unobserved
heterogeneity ${\bf A}$ \textcolor{black}{unaffected by this exogenous
change held} \emph{fixed}.}

\textcolor{black}{To achieve this separation, we seek to identify
and estimate the integrand $\E\left[\rest{y_{ijt}}{\bf X}_{it}=\ol{{\bf X}},{\bf A}_{i}\right]$,
which is a function of $\b_{0}$. Conditional on $\bm{\d}_{it}=\ol{\X}^{'}\b_{0}=:\ol{\bm{\d}}$
and ${\bf A}_{i}$, by our model specification (\ref{eq:Model_PMC})
we have
\[
\E\left[\rest{y_{ijt}}\bm{\d}_{it}=\ol{\bm{\d}},{\bf A}_{i}\right]=\psi_{j}\left(\ol{\bm{\d}},\ {\bf A}_{i}\right)=:\psi_{ij}\left(\ol{\bm{\d}}\right).
\]
Here, an important observation is that although the individual heterogeneity
${\bf A}_{i}$ is not directly observable, the identity of $i$ is
observable. We can hold individual $i$ fixed in the regression to
control for ${\bf A}_{i}$ and only use variations in the data across
$t$ in a long panel setting to estimate the conditional choice probability.
Specifically, suppose we have long panels, i.e., $T\rightarrow\infty$.
Suppose the conditions in Appendix \ref{sec:App_PID} are also satisfied
so that $\d_{0}$ is point identified. We can identify $\psi_{ij}\text{\ensuremath{\left(\ol{\bm{\d}}\right)}}$
for each fixed $i$ and product $j$ and estimate it via nonparametric
regression of $y_{ijt}$ on $\text{\ensuremath{\left(\d_{i1t},\ldots,\d_{iJt}\right)}}$
across $t=1,\ldots,T$, under suitable stationarity and weak dependence
conditions on the error terms. See, for example, \citet*{su2012sieve}
for a sieve estimator of functional fixed effects in a long panel
setting.}

We are now in the position to estimate the counterfactual \textcolor{black}{market
shares of product $j=1,\ldots,J$ at any counterfactual $\ol{{\bf X}}$,
where $\ol{{\bf X}}$ is a $D\times J$ matrix. First, we use the
estimated $\hat{\b}$ to compute the counterfactual index $\hat{\bm{\d}}\left(\cd\right)$
evaluated at $\ol{{\bf X}}$, i.e., $\hat{{\bf \bm{\d}}}\left({\bf \ol X}\right)=\ol{\X}^{'}\hat{\b}.$
Then, we obtain $\hat{\psi}_{ij}\left(\hat{{\bf \bm{\d}}}\left({\bf \ol X}\right)\right)$
by plugging $\hat{{\bf \bm{\d}}}\left({\bf \ol X}\right)$ into the
nonparametric estimate $\hat{\psi}_{ij}\left(\cd\right)$ for each
fixed $i$ and product $j$. Finally, we estimate $\mu_{j}\left(\ol{{\bf X}}\right)$
for product $j$ by averaging over individuals in the sample, i.e.,
$\mu_{j}\left(\ol{{\bf X}}\right)=\frac{1}{N}\sum_{i=1}^{N}\hat{\psi}_{ij}\left(\hat{{\bf \bm{\d}}}\left({\bf \ol X}\right)\right)$.
The last step corresponds to the integration with respect to the probability
measure of $A$ in (}\ref{eq:cf_mu_j}\textcolor{black}{).}

\section{\label{sec:App_ADDSIM}Additional Simulation Results}

\subsection{Adaptive-Grid Computation Algorithm}

In this section, we illustrate a typical output of our second-step
computation algorithm based on the adaptive-grid search over the angle
space, and show that the algorithm works well. For this purpose we
consider a simplified DGP without fixed effect $A_{ij}$. We draw
each of $X_{ijt,d}$ independently across each dimension $d$ from
the standard normal distribution, and set the distribution of the
idiosyncratic shock to be $\epsilon_{ijt}\sim TIEV\left(0,1\right)$,
so that we can skip the first-step estimation and directly compute
the true conditional choice probability. Note that the conditions
for point identification of $\b_{0}$ are satisfied. Because we are
only seeking to illustrate the validity of the algorithm itself, we
set $N$ to be large with $N=10^{7}$ and fix $D=3,\ J=3,\ \text{and }T=2$.
Then, we apply our adaptive-grid algorithm to search for $\beta_{0}$.
\begin{figure}
\begin{centering}
\includegraphics[scale=0.7]{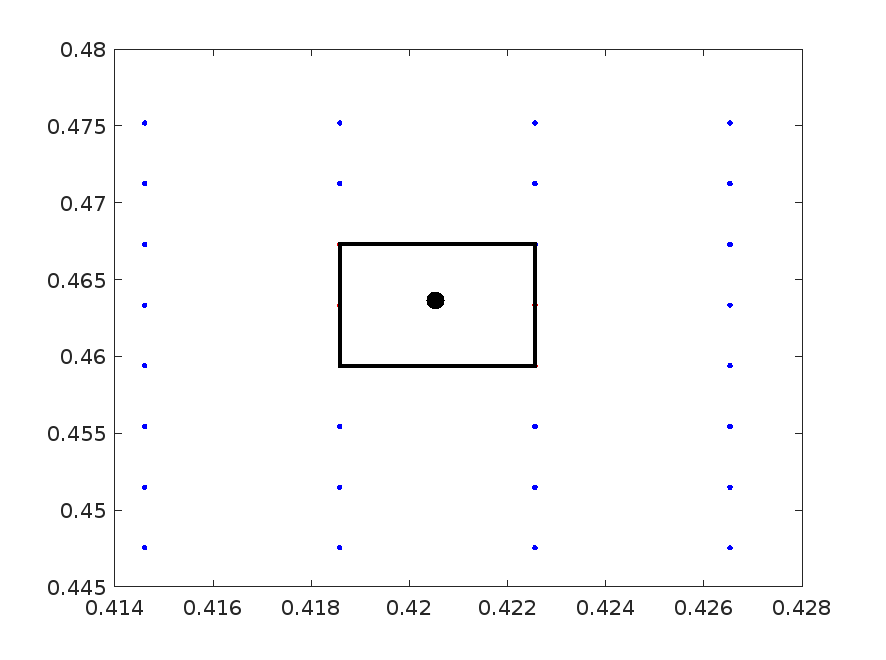}
\par\end{centering}
\caption{The Argmin Set in $\protect\T$\label{fig:OracleTheta}}
\end{figure}

Figure \ref{fig:OracleTheta} shows how our computational algorithm
works in finding the true unknown $\theta_{0}$, the angle representation
of the true $\b_{0}$ in the $\Theta$ space. The horizontal and vertical
axes correspond to the two polar coordinates that are associated with
$\S^{2}$. The blue dots represent the points that our algorithm searches
over but find \textit{not} to be minimizers of the sample criterion
$\hat{Q}$. The black box indicates the area that the minimizers for
the sample criterion $\hat{Q}$ lie within, or more precisely, a rectangular
enclosure of the numerical argmin set. The big black dot stands for
the true parameter value $\theta_{0}=\left(0.4205,0.4636\right)^{'}$.

It is evident from Figure \ref{fig:OracleTheta} that our adaptive-grid
algorithm is able to correctly locate an area that covers the true
$\theta_{0}$, which lies within the small black box representing
the estimated set of $\hat{\theta}$, demonstrating the efficacy of
the algorithm. Besides, it is worth mentioning that our algorithm
computes reasonably fast, as it first performs a rough search on the
whole unit sphere $\S^{2}$, then focuses on the area where the minimizers
are most likely to lie. In the last few rounds of search, the algorithm
evaluates the criterion function $\hat{Q}$ on a relatively small
area of points shown by those blue and red dots in Figure \ref{fig:OracleTheta}
until the desired level of accuracy is achieved.

For a more transparent representation, we translate the angles $\theta$
in the polar coordinates into unit vectors $\b$ on the unit sphere
$\S^{2}$ and show it in Figure \ref{fig:OracleBeta}, which is now
plotted on $\S^{2}\subseteq\R^{3}$. Again the blue dots represent
the points that do not achieve the minimum of $\hat{Q}$; the black
box shows an enclosing set of the minimizers of $\hat{Q}$. The big
black dot represents the true parameter value $\b_{0}$, which resides
inside the black box of the minimizers of $\hat{Q}$. Figure \ref{fig:OracleBeta}
illustrates that our computation algorithm is able to locate a tight
area around $\b_{0}$.
\begin{figure}
\begin{centering}
\includegraphics[scale=0.7]{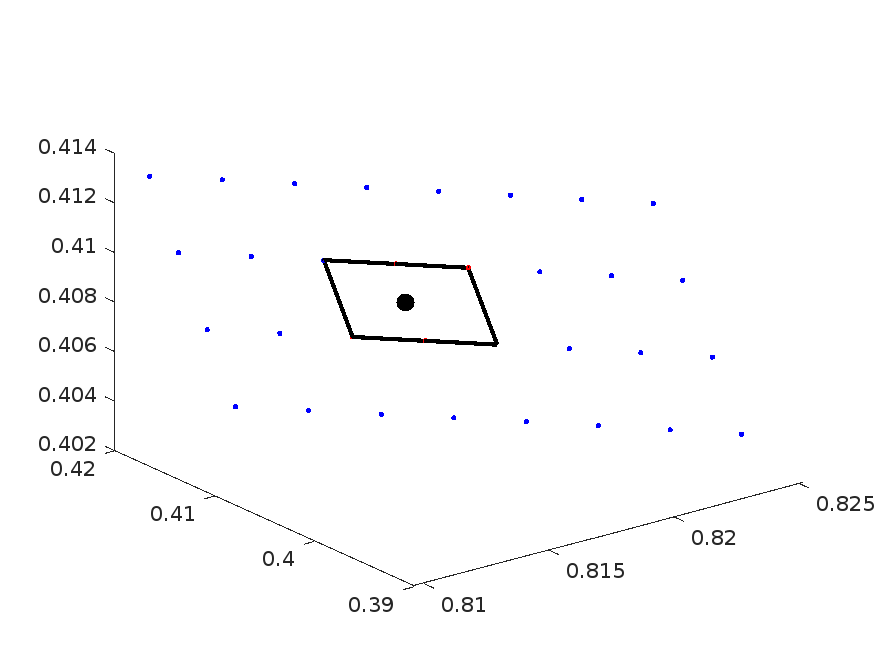}
\par\end{centering}
\caption{The Argmin Set in $\protect\S^{2}$\label{fig:OracleBeta}}
\end{figure}

\subsection{Estimation without Point Identification}

We now investigate the performance of our estimator when point identification
fails. To make things comparable, we fix $\left(N,D,J,T\right)$ as
in the baseline case, but modify the configuration in two different
ways. We maintain the point identification in one setting but lose
the point identification in the other. Specifically, we set $Z_{i}\sim\mathcal{U}\left[-\sqrt{3},\sqrt{3}\right]$,
$X_{ijt,1}\sim\mathcal{U}\left[-1,1\right]$, $X_{ijt,2}=Z_{i}+\mathcal{N}\left(0,6\right)$,
and $X_{ijt,3}\sim\mathcal{N}\left(0,1\right)$ for the point identified
case. For the DGP without point identification, we let $Z_{i}\sim\mathcal{U}\left[-\sqrt{3},\sqrt{3}\right]$,
$P(X_{ijt,1}=1/2)=P(X_{ijt,1}=-1/2)=0.5$, $X_{ijt,2}=Z_{i}+\mathcal{U}\left[-\sqrt{6},\sqrt{6}\right]$,
and $P(X_{ijt,3}=1/2)=P(X_{ijt,3}=-1/2)=0.5$. The construction of
$A_{ij}$ is the same as in the DGP for the baseline results in Table
\ref{tab:Baseline-Estimation-Performance}. We deliberately control
the location and scale of each variable to be comparable across the
two configurations, with the only differences being the presence of
discreteness and boundedness of supports. When point identification
fails, we compute the set estimator $\hat{B}_{\hat{c}}$ of \eqref{eq:Theta_hat_CHT07}
with $\hat{c}>0$. Table \ref{tab:PerfIDFurther} contains simulation
results under the two configurations, with different choices of $\hat{c}$
when point identification fails.\footnote{Specifically, noting that $c_{N}\log N\leq N^{-1/4}\log N\approx0.92\leq1$
for $N=10,000$, we set $\hat{c}=0.01,$ $0.1$ and $1$, respectively.}

\begin{table}
\caption{Performance with and without Point ID: Further Examination \label{tab:PerfIDFurther}}

\bigskip{}

\noindent \centering{}%
\begin{tabular}{cccccccc}
\toprule 
\multirow{2}{*}{point ID ?} & \multirow{2}{*}{$\hat{c}$} & \multicolumn{3}{c}{$\phantom{\frac{\frac{1}{1}}{\frac{1}{1}}}$rMSE} & \multicolumn{3}{c}{MND}\tabularnewline
\cmidrule{3-8} \cmidrule{4-8} \cmidrule{5-8} \cmidrule{6-8} \cmidrule{7-8} \cmidrule{8-8} 
 &  & $\phantom{\frac{\frac{1}{1}}{\frac{1}{1}}}$$\hat{\b}^{m}$ & $\hat{\b}^{u}$ & $\hat{\b}^{l}$ & $\hat{\b}^{m}$ & $\hat{\b}^{u}$ & $\hat{\b}^{l}$\tabularnewline
\midrule
$\phantom{\frac{\frac{1}{1}}{\frac{1}{1}}}$ \textbf{(i)} yes & $\phantom{\frac{\frac{1}{1}}{\frac{1}{1}}}$- & 0.0574 & 0.0636 & 0.0665 & 0.0501 & 0.0577 & 0.0602\tabularnewline
\midrule
\multirow{3}{*}{$\phantom{\frac{\frac{1}{1}}{\frac{1}{1}}}$\textbf{(ii)} no} & $\phantom{\frac{\frac{1}{1}}{\frac{1}{1}}}$0.01 & 0.0624 & 0.0657 & 0.0711 & 0.0536 & 0.0590 & 0.0630\tabularnewline
 & $\phantom{\frac{\frac{1}{1}}{\frac{1}{1}}}$0.1 & 0.0521 & 0.0820 & 0.0966 & 0.0444 & 0.0782 & 0.0904\tabularnewline
 & $\phantom{\frac{\frac{1}{1}}{\frac{1}{1}}}$1 & 0.0476 & 0.2198 & 0.2631 & 0.0458 & 0.2190 & 0.2619\tabularnewline
\bottomrule
\end{tabular}
\end{table}

Across rows in (i) and (ii), we see that the lack of point identification
does negatively affect the performance of our estimates, but the impact
is limited to a moderate degree. Within rows in (ii), we observe that,
as expected, a more conservative choice of the constant $\hat{c}$
worsens performance of the upper and lower bounds by enlarging the
estimated sets; meanwhile, it appears that the size and the performance
of our estimator based on $\hat{\b}^{m}$ is not sensitive to the
choice of $\hat{c}$.

\subsection{Results Varying $D,J,T$}

In this section, we show how our estimator performs under different
$\left(D,J,T\right)$. We maintain $N=10,000$ as in the baseline
configuration. We draw $Z_{i}\sim\mathcal{N}\left(0,1\right)$ and
construct $A$ and $X$ according to the following specifications:
\begin{align*}
A_{ij}\sim\begin{cases}
0, & j=1,\\
\left[Z_{i}\right]_{+}, & j=2,\\
\mathcal{U}\left[-0.25,0.25\right], & j=3,\ldots,J,
\end{cases} & X_{ijt,d}\sim\begin{cases}
U\left[-1,1\right], & d=1,\\
Z_{i}+\cN\left(0,6\right), & d=2,\\
\cN\left(0,1\right), & d=3,\ldots,D,
\end{cases}
\end{align*}
which coincides with the baseline model at $D=3,\ J=3$. We emphasize
that in all configurations we allow for nonlinear dependence between
$A$ and $X$ via $Z$.

We report in Table \ref{tab:PerfVaryDJT} the performance of our estimators
for each of the corresponding configurations across all $M=1,000$
simulations.

\begin{table}[h]
\caption{Performance Varying $D,J,T$\label{tab:PerfVaryDJT}}

\bigskip{}

\centering{}%
\begin{tabular}{ccccc}
\toprule 
\multirow{2}{*}{rMSE} & \multicolumn{2}{c}{$\phantom{\frac{\frac{1}{1}}{\frac{1}{1}}}$ $J=3$} & \multicolumn{2}{c}{$\phantom{\frac{\frac{1}{1}}{\frac{1}{1}}}$ $J=4$}\tabularnewline
\cmidrule{2-5} \cmidrule{3-5} \cmidrule{4-5} \cmidrule{5-5} 
 & $\phantom{\frac{\frac{1}{1}}{\frac{1}{1}}}$ $T=2$ & $T=4$ & $\phantom{\frac{\frac{1}{1}}{\frac{1}{1}}}$ $T=2$ & $T=4$\tabularnewline
\midrule
$D=3$ & 0.0594 & 0.0411 & 0.1321 & 0.0982\tabularnewline
$D=4$ & 0.0725 & 0.0497 & 0.1384 & 0.1091\tabularnewline
\midrule
\multirow{2}{*}{MND} & \multicolumn{2}{c}{$\phantom{\frac{\frac{1}{1}}{\frac{1}{1}}}$ $J=3$} & \multicolumn{2}{c}{$\phantom{\frac{\frac{1}{1}}{\frac{1}{1}}}$ $J=4$}\tabularnewline
\cmidrule{2-5} \cmidrule{3-5} \cmidrule{4-5} \cmidrule{5-5} 
 & $\phantom{\frac{\frac{1}{1}}{\frac{1}{1}}}$ $T=2$ & $T=4$ & $\phantom{\frac{\frac{1}{1}}{\frac{1}{1}}}$ $T=2$ & $T=4$\tabularnewline
\midrule
$D=3$ & 0.0522 & 0.0363 & 0.1154 & 0.0860\tabularnewline
$D=4$ & 0.0660 & 0.0475 & 0.1264 & 0.1003\tabularnewline
\bottomrule
\end{tabular}
\end{table}

From Table \ref{tab:PerfVaryDJT} we find a larger $T$ improves the
performance of our estimator, which is arguably more practically relevant
given the increasing availability of long panel data. The improvement
in performance with larger $T$ is because our method can extract
more information from $T\times\left(T-1\right)$ ordered pairs of
time periods which effectively increase the total number of observations.
We also find that increase in $D$ or $J$ adversely affects the performance
of our estimator, which is expected because more information is required
to estimate more covariates or deal with more alternatives. For example,
when $J$ is 3 and $T$ is 4, an increase in the dimension of product
characteristics $D$ from 3 to 4 increases the rMSE from 0.0411 to
0.0497. The change in performance for increasing $J$ is more significant.
For instance, when $D=4$ and $T=4$, an increase in $J$ from 3 to
4 increases the MND from 0.0475 to 0.1003.

\subsection{Robustness Against the Blue-Bus/Red-Bus Problem\label{subsec:Blue-Bus/Red-Bus}}

We now demonstrate the robustness of our method against the Blue-Bus/Red-Bus
problem, or more precisely, a case where the product set includes
highly substitutable alternatives. 

Specifically, we modify the DGP in our baseline simulation by increasing
$J$ from 3 to 4, and introduce the fourth product as an indistinguishable
duplicate of the third one: they share the same observable and unobservable
characteristics, including the observable characteristics $X$, unobserved
fixed effect $A$, and the unobserved error term $\e$. We equally
split the market share between $j=3$ and $j=4$. 

Using the above DGP, we compare our method's performance to that of
the RCLM estimator. For the RCLM estimator, we simulate $M=1,000$
random coefficients $\b_{i}=\b+\eta_{i}$ where $\eta_{i}\sim N\left(0,I_{D}\right)$
for the consumers with heterogeneous preferences, where $I_{D}$ is
a $D\times D$ identity matrix. The predicted market shares are thus
computed as 
\begin{align}
\s_{jt}\left(X_{i\cd t};\b\right) & =\E\left[\rest{\ind\left\{ U_{ijt}>\max_{k\neq j}U_{ikt}\right\} }X_{i\cd t}\right]\nonumber \\
 & =\int\dfrac{\exp\left(X_{ijt}^{'}\b+X_{ijt}^{'}\eta_{i}\right)}{\sum_{k=1}^{J}\exp\left(X_{ikt}^{'}\b+X_{ikt}^{'}\eta_{i}\right)}dF_{\eta}\left(\eta_{i}\right)\nonumber \\
 & \simeq\dfrac{1}{M}\sum_{m=1}^{M}\dfrac{\exp\left(X_{ijt}^{'}\b+X_{ijt}^{'}\eta_{i}^{\left(m\right)}\right)}{\sum_{k=1}^{J}\exp\left(X_{ikt}^{'}\b+X_{ikt}^{'}\eta_{i}^{\left(m\right)}\right)}.\label{eq:rclm eq}
\end{align}
For each round of simulation, we compute predicted market shares by
numerically integrating over the random taste distribution over $M$
agents. We then choose the candidate $\b$ over a fine grid of $L=10,000$
points from $\S^{D-1}$ that achieves the smallest distance in the
Euclidean norm between the predicted and observed market shares as
$\hat{\b}$. We use $N=1,000$, $D=3$, $J=4$, $T=2$, and $B=$1,000
for this simulation exercise. 

Table \ref{tab:blue-bus-red-bus} summarizes the results. The sum
of absolute biases of $\hat{\b}$ under our method is 0.0452, which
is markedly lower than that of the RCLM estimator (0.8689). The average
size of the estimated sets is 0.0585 for our method. Since, by construction,
the RCLM estimator is a point estimator, its average size is zero
across simulations. Finally, the rMSE and MND are significantly lower
for our estimator than for the RCLM. The results show that, in the
presence of highly similar products (with highly correlated unobserved
heterogeneity), the RCLM estimator cannot recover the structural parameters
accurately. In contrast, our method is robust to arbitrary dependence
structures in the unobserved heterogeneity terms across products,
leading to more accurate estimates of the structural parameters.

\begin{table}
\caption{Illustration of the \textquotedblleft Blue-Bus/Red-Bus\textquotedblright{}
Problem\label{tab:blue-bus-red-bus}}

\bigskip{}

\noindent \centering{}%
\begin{tabular}{ccccc}
\toprule 
\multirow{2}{*}{} & \multirow{2}{*}{$\sum_{d}\left|\text{bias}_{d}\right|$} & \multirow{2}{*}{$\sum_{d}\text{mean(u-l)}_{d}$} & \multirow{2}{*}{\textcolor{black}{$\text{rMSE}$}} & \multirow{2}{*}{MND}\tabularnewline
 &  &  &  & \tabularnewline
\midrule
Current Paper & 0.0452 & 0.0585 & \textcolor{black}{0.1389} & 0.1210\tabularnewline
RCLM & 0.8689 & - & \textcolor{black}{0.6047} & 0.6042\tabularnewline
\bottomrule
\end{tabular}
\end{table}

\end{document}